\documentclass[12pt, a4paper]{article}
\usepackage[english]{babel}

\usepackage{setspace}
\usepackage[a4paper, left=1in, right=1in, top=1.5in, bottom=1.5in]{geometry}

\usepackage{bbm}
\usepackage{amsmath, amsthm, amssymb}
\usepackage{dsfont}
\usepackage{float}
\usepackage{tikz}
\usetikzlibrary {arrows.meta,bending,positioning}
\usetikzlibrary{shapes,arrows,chains}

\usepackage{xcolor} 
\usepackage[colorlinks = true, 
urlcolor = black, 
linkcolor = black, 
citecolor = black]{hyperref}

\usepackage{listings}
\newtheorem{theorem}{Theorem}
\newtheorem{lemma}{Lemma}
\newtheorem{assumption}{Assumption}
\newtheorem{definition}{Definition}

\usepackage{graphicx}
\usepackage{booktabs}
\usepackage{tabularx}
\usepackage{ltxtable}
\usepackage{placeins}
\usepackage{changepage}
\usepackage{pdfpages}

\usepackage[labelfont={bf}, textfont={bf}, justification=raggedright]{caption}
\usepackage[textfont=footnotesize, justification=justified]{subcaption} 

\usepackage[round]{natbib}
\usepackage{titlesec}
\makeatletter
\def\@seccntformat#1{\@ifundefined{#1@cntformat}%
   {\csname the#1\endcsname\space}
   {\csname #1@cntformat\endcsname}}
\newcommand\section@cntformat{\thesection.\space} 
\makeatother

\newcolumntype{Y}{>{\centering\arraybackslash}X}
 
\def\sym#1{\ifmmode^{#1}\else\(^{#1}\)\fi}

\usepackage{pifont}
\newcommand{\cmark}{\ding{51}}
\newcommand{\xmark}{\ding{55}}

\begin{document}

	\title{\textbf{Building Trust Takes Time: \\ {\Large Limits to Arbitrage for Blockchain-Based Assets}}\thanks{Hautsch is at University of Vienna -- Department of Statistics and Operations Research, Research Platform Data Science @ Uni Vienna, Vienna Graduate School of Finance (VGSF) and Center for Financial Studies (CFS). Scheuch is at wikifolio Financial Technologies AG, and Voigt is at the University of Copenhagen -- Department of Economics. Voigt gratefully acknowledges support from the Danish Finance Institute (DFI). Corresponding author: \href{mailto:stefan.voigt@econ.ku.dk}{stefan.voigt@econ.ku.dk}. We thank Christine Parlour, one anonymous referee, Bruno Biais, Sylvia Fr\"uhwirth-Schnatter, Sergey Ivliev, David Easley, Katya Malinova, Fahad Saleh, Peter Zimmerman, Viktor Todorov, Majeed Simaan, Jun Aoyagi as well as seminar participants at QFFE 2018, the $1^{st}$ International Conference on Data Science in Finance with R, the $4^{th}$ Konstanz-Lancaster Workshop on Finance and Econometrics, the Crypto Valley Blockchain Conference 2018, HFFE 2018, CFE 2018, CUNEF, University of Heidelberg, University of Vienna, University of Graz, the $2^{nd}$ Toronto FinTech Conference, the $4^{th}$ Vienna Workshop on High-Dimensional Time Series 2019, the Conference on Market Microstructure and High Frequency Data 2019, the 2019 FIRS Conference, the $12^{th}$ Annual SoFiE Conference, the IMS at the National University of Singapore, the $3^{rd}$ SAFE Microstructure Conference, the 2019 EFA Annual Meeting, the 2019 Vienna Congress on Mathematical Finance, the International Conference on Fintech \& Financial Data Science 2019, the $4^{th}$ International Workshop in Financial Econometrics, the 2019 CFM-Imperial Workshop, the Western Finance Association Conference 2020, the NHH Young Scholars Finance Webinar series, the European Blockchain Center, the Conference on Financial Innovation at Stevens Institute for Technology, the 4th UWA Blockchain and Cryptocurrency Conference and at U.S. Securities and Exchange Commission for helpful comments and suggestions. This paper replaces an earlier draft titled "Limits to Arbitrage in Markets with Stochastic Settlement Latency."}
	}
	\author{\begin{tabularx}{1\textwidth}{YYY}
			  Nikolaus Hautsch & Christoph Scheuch & Stefan Voigt  
	\end{tabularx}
	}

\date{\today \vspace{-2.5em}}

	\maketitle
        \thispagestyle{empty}
	\renewcommand{\abstractname}{\vspace{-\baselineskip}}
	\begin{abstract}
	\begin{spacing}{1}
          \noindent 
A blockchain replaces central counterparties with time-consuming consensus protocols to record the transfer of ownership. 
This settlement latency slows cross-exchange trading, exposing arbitrageurs to price risk. 
Off-chain settlement, instead, exposes arbitrageurs to costly default risk. 
We show with Bitcoin network and order book data that cross-exchange price differences coincide with periods of high settlement latency, asset flows chase arbitrage opportunities, and price differences across exchanges with low default risk are smaller.
Blockchain-based trading thus faces a dilemma: Reliable consensus protocols require time-consuming settlement latency, leading to arbitrage limits. Circumventing such arbitrage costs is possible \emph{only} by reinstalling trusted intermediation, which mitigates default risk.
        \end{spacing}
  
		\vspace{12pt}
	\end{abstract}

\newpage
\setcounter{page}{1}	

\section{Introduction}\label{sec:introduction}

Whenever two investors seek to agree on a financial contract, counterparty risk harms trading if either party may default on its contractual obligations. 
To guarantee the execution of the negotiated terms, traditional stock markets organize the trading process around trusted intermediaries. 
Typically, central clearing counterparties bear all counterparty risk between transacting parties from the trade agreement until the legal transfer of ownership through security depositories. Market participants pay for the implied insurance against counterparty risk through fees and collateral deposits.

By contrast, blockchain technology promises to mitigate counterparty risk and render trusted intermediaries obsolete.
In such a framework, an open network of validators establishes consensus about transaction histories. 
Consensus protocols serve as the regulatory framework and specify how validators reach an agreement and are incentivized to collaborate. 
The design of consensus protocols can take different forms to ensure a reliable record of transaction histories. 
One widely applied consensus mechanism is proof-of-work, where the validation process involves substantial computational effort and is hence costly to undermine \citep[e.g.,][]{Biais.2019}. 

The design of exchanges as facilitators of the transfer of blockchain-based assets between investors can take different forms. 
Decentralized exchanges (DEX) enable peer-to-peer trading and thus do not require trust in an intermediary. 
DEXes are essentially smart-contract-based algorithms that record every transaction directly on the blockchain and render exchanges pure matchmakers \citep[e.g.,][]{Harvey.2020, Lehar.2021}. 
While entirely mitigating counterparty risk, this market structure is not appropriate for high-frequency trading, hindering the widespread adoption of DEXes. 
For example, consider running a traditional limit order book: As each order submission requires blockchain validation, which involves a fee to validators, a large message volume renders DEXes comparably slow, inefficient, and costly.

In fact, most transactions of crypto tokens rely on fragmented, largely unregulated, and so-called centralized exchanges (CEX).\footnote{DEX trading volume is negligible relative to CEX volume. According to data provider \emph{cryptocompare}, CEXes executed 89\% of digital asset volume of approximately USD 1.04 trillion in December 2021. Serious concerns regarding the limits to DEX adoption exist. Among others, \cite{Capponi.2021} and \cite{Park.2021} illustrate that DEXes render liquidity provisioning and, subsequently, trading on private information prohibitively costly.} 
A CEX facilitates trades, typically by running a limit order book maintained internally and subsequently settling transactions off-chain. 
Off-chain settlement means that CEXes net executed transactions internally such that transactions are not validated on the blockchain, and validators need not be compensated for executed order flow. To render off-chain settlement feasible, however, CEXes need to serve as custodians of their customers' funds. 
Moreover, market fragmentation across multiple CEXes can result in violations of the law of price. 
Thus, risk and marginal costs for cross-exchange arbitrageurs determine the degree to which such violations of the law of one price can persist. 

The main result of this paper is that \emph{settlement latency} exposes cross-exchange arbitrageurs to substantial costs. Settlement latency is the waiting time until blockchain validation of a transaction. 
Settlement latency makes arbitrage trades costly since it exposes arbitrageurs to substantial price risk: 
Due to the absence of a trusted intermediary, exploiting price differences requires the transfer of a blockchain-based asset \emph{across} CEXes and, therefore, time-consuming blockchain validation. 

Consider the decision problem of a risk-averse cross-exchange arbitrageur who monitors the prices of a blockchain-based asset on two CEXes. 
Whenever she buys on one exchange that quotes a lower price (low-price exchange), she must wait until blockchain validation of the asset transfer before selling on the other exchange (high-price exchange). 
Thus, the settlement latency underlying this transfer exposes the arbitrageur to the risk of adverse price movements. 
Consequently, risk-averse arbitrageurs exploit concurrent price differences only if these price differences are sufficiently large to compensate for the price risk due to settlement latency.

We derive a closed-form expression for the arbitrageur's certainty equivalent and show that it increases with (i)  the expected settlement latency, (ii) the variance of the settlement latency, (iii)  price volatility on the high-price CEX, and (iv) the arbitrageur’s risk aversion. Our characterization of the certainty equivalent also accounts for transaction costs and optimally chosen settlement fees which incentivize validators to enable faster validation \citep[e.g.,][]{Easley.2017, Huberman.2021, Lehar.2022}.

Settlement latency can be effectively circumvented \emph{only} if the arbitrageur deposits arbitrage capital as collateral of the blockchain-based asset under the custody of CEXes. 
While capital intensive, the arbitrageur can acquire a marginal unit at the low-price exchange and immediately dispose of an offsetting amount of her inventory on the high-price exchange. 
With such inventory, settlement latency does not affect the trading decision. 
However, off-chain settlement exposes traders to exchange default risk, which manifests in the risk of thefts, hacks, or exit scams. For instance, \cite{Biais.2019b} document more than 50 hacks and other losses on Bitcoin exchanges and find that the expected returns of cryptocurrency investors reflect these risks.

As a result, the dilemma of blockchain-based settlement unfolds as follows: Settlement without trusted intermediation requires a secure validation system. However, a reliable consensus protocol relies on sufficiently high settlement latency \citep[e.g.,][]{Hinzen.2022}. 
Marginal costs for arbitrageurs on CEXes are high when settlement latency is large. Ergo:
High settlement latency increases trading costs and improves the attractiveness of investments into CEX reliability. We may thus expect that successful blockchain-based settlement implies higher efforts to improve the trust in CEXes.
As long as CEXes struggle to establish themselves as trusted intermediaries, default risk renders settlement latency an inevitable friction hampering arbitrage activity.\footnote{Establishing trust requires costly investments in security, governance, and regulatory compliance. 
Indeed, modern CEXes exert substantial effort for such trust-enhancing procedures, i.e., by increasing transparency of cryptocurrency holdings under custody, implementing funds to repay damages due to security breaches, or complying with strict regulatory standards, e.g., the SEC BitLicense. Changpeng Zhao, the CEO of Binance, the world's largest CEX, expressed this concern: ''By focusing on delivering a superior user experience in tandem with top-notch security, centralized exchanges essentially live or die based on their ability to “create trust” among their users.''} 

Our analysis yields two major empirical predictions: (i) settlement latency is a costly friction for cross-exchange arbitrageurs, and (ii) mitigating CEX default risk reduces marginal cross-exchange arbitrage costs. Our empirical analysis relies on novel, granular data to provide compelling evidence for both effects at play.
We gather minute-level data from the order books of 16 large CEXes that feature trading Bitcoin against the US Dollar (USD) between January 2018 and October 2019. 
We analyze violations of the law of one price between each exchange pair. In line with \cite{Kroeger.2017}, \citet{Choi.2018}, \cite{Makarow.2018}, and \cite{Borri.2021}  we report substantial price differences across the 120 feasible exchange pairs through our sample.
To quantify the relevance of settlement latency as a friction, we enrich our dataset with comprehensive high-frequency information about the Bitcoin network, which includes the settlement latency, i.e., the time it takes for every transaction from entering the Bitcoin network until its inclusion in the blockchain. 

In line with our theoretical framework, we find that large price differences coincide with periods of high settlement latency, high latency uncertainty, and high spot volatility. 
We find that substantial uncertainty in settlement latency due to the computational effort in the proof-of-work mechanism of Bitcoin contributes more than 40\% of the marginal arbitrage costs. 
The results are robust when we control for trading costs, optimal validator fee choices, and order book liquidity, in line with \cite{Roll.2007} and \cite{DeJong.1990}. We explicitly distinguish between settlement latency and trading-related frictions such as tight capital controls. For instance, \cite{Makarow.2018} and \citet{Choi.2018} argue that capital controls create frictions that limit the ability of arbitrageurs to take advantage of persistent price differences. In our sample, settlement latency remains relevant even if we restrict the analysis to CEXes that are active within the same region, thus eliminating the potential impact of capital controls.
We document further that these results hold even beyond other well-known frictions that hamper arbitrage activity, such as the presence of intermediation facilities, e.g., margin trading \cite[e.g.,][]{Pontiff.1996, Lamont.2003b, Lamont.2003, Jong.2009}.

To analyze the role of exchange default risk, we use the number of Bitcoins under the custody of a CEX as a proxy for trust. This proxy relies on our theoretical analysis of arbitrage capital allocation in the presence of settlement latency \emph{and} default risk: CEXes that successfully mitigate default risk should attract more arbitrage capital.
Our empirical analysis reveals a strong increase in Bitcoin holdings under the custody of CEXes during our sample period. In October 2019, CEXes held more than USD 12.4 Billion of Bitcoin, an increase of 25.8\% relative to the beginning of our sample period in January 2018.

Consistent with our theoretical predictions, cross-exchange price differences tend to be narrower between CEXes with more funds under custody. 
Simultaneously, settlement latency remains a significant driver of the magnitudes of observed cross-exchange price differences and their variation over time, even after controlling for funds under custody as a proxy for trust. Violations of the law of one price hence coincide with periods of high settlement latency, even for CEXes perceived as highly trustworthy.\footnote{We confirm that exchanges with a large number of Bitcoins under their custody are associated with higher exchange ratings as an alternative measure for trust. Data provider \emph{cryptocompare} classifies CEXes according to implemented due diligence procedures, operational transparency, security, and data provision. Exchanges with a higher rating are more likely to exhibit larger funds under their custody. Our results remain qualitatively unchanged if we measure trust via exchange ratings instead of funds under custody.}  

Finally, we provide evidence for cross-exchange asset flows chasing price differences but being hampered by settlement latency. 
We collect Bitcoin wallet IDs under the control of the CEXes in our sample and compile a unique and novel data set of 3.9 million cross-exchange transactions with an average daily volume of USD 72 million. We use an instrument for cross-exchange price differences to tackle the inherent endogeneity arising from the simultaneity between price differences and cross-exchange asset flows. For that purpose, our instrument is the theoretical minimum price difference necessary such that the arbitrageur prefers to trade. 
We estimate these \emph{arbitrage bounds} based on our granular Bitcoin order book and network data. 
Asset flows into an exchange significantly respond to variations in concurrent price differences, particularly those explained by variations in arbitrage bounds, while also controlling for settlement latency-induced price risk and exchange-specific characteristics. 

In summary, our empirical results indicate that (i) violations of the law of one price coincide with periods of high settlement latency, high settlement latency uncertainty, and spot volatility, (ii) arbitrageurs perceive settlement latency as an economically relevant friction and act in anticipation of the related price risk and (iii) trust mitigates these frictions. In this sense, our paper contributes to the literature on limits to arbitrage \cite[e.g.,][]{DeJong.1990, Shleifer.1997, Gromb.2010} by highlighting a friction that arises specifically for blockchain-based assets. 
Our results contribute to a better understanding of the economic implications of blockchain technologies for trading on financial markets. 
The promise of fast and low-cost transaction settlement leads central banks and marketplaces to actively explore potential applications of such systems for transaction settlement \citep[e.g.,][]{BIS.2017, NASDAQ.2017, ECB.2020}. 
The existing literature on blockchain-based trading mainly focuses on the incentive compatibility constraints of \emph{validators}, which limit decentralization. In particular, \cite{Brunnermeier.2018} points out a "Blockchain dilemma" in the sense that correctness, decentralization, and cost efficiency cannot simultaneously be achieved for blockchain-based trading. Along this line, \cite{Cong.2020} show that risk-averse validators have incentives to pool their mining power which can result in inefficient accumulations of mining capacities. \cite{Hinzen.2022} show that network security and fast settlement are mutually exclusive and \cite{Pagnotta.2018} finds that miners' competition for block rewards amplifies price volatility and thus has severe pricing implications.

Our focus lies on the economic frictions that originate from the time-consuming effort necessary to mitigate counterparty risk for blockchain-based assets. 
Materialization of such counterparty risk in cryptocurrency exchanges has been documented in various forms, for instance, price manipulations \citep[e.g.,][]{Gandal.2018, Li.2021} or wash trading \citep[e.g.,][]{Aloosh.2019, Amiram.2020, Cong.2021}.
We thus shed light on the important open question: to which extent can consensus protocols maintain reliable and secure validation \emph{and} minimal third-party intermediation? Our main conclusion is that settlement latency is a severe economic friction and can be bypassed only by reinstalling core components of trusted intermediation with all associated costs or risks.

Our results are also relevant from the perspective of platform spillovers for blockchain-based assets. 
For instance, \cite{Capponi.2021} shows that DEX arbitrage competition can lead to elevated settlement fees, and \cite{Sokolov.2021} finds that periods of elevated Ransomware activity increase the expected settlement latency. 
The results in our paper thus predict higher marginal arbitrage costs due to events that are isomorphic to cross-exchange trading but nevertheless resemble negative externalities on \emph{all} validation activities conducted on the blockchain, including cross-exchange trading.    

The paper proceeds as follows. In Section~\ref{sec:theory}, we derive arbitrage costs due to settlement latency and CEX default risk and characterize the resulting limits to arbitrage capital. Section~\ref{sec:methodology} provides a general theoretical framework to quantify latency-related costs in a realistic setup where the arbitrageur endogenizes transaction costs, random settlement latency, optimal trading quantities, as well as settlement fee choices. In Section~\ref{sec:data}, we introduce the Bitcoin order book and network data. In Section~\ref{sec: results}, we quantify the core components of our theoretical framework, settlement latency, and spot volatility. In Section~\ref{sec:regressions}, we provide empirical results and show that settlement latency is an economically relevant friction for blockchain-based assets. Section~\ref{sec:conclusion} concludes. 

\section{Limits to Arbitrage for Blockchain-Based Assets}\label{sec:theory}

\begin{figure}
\caption{Illustration of the cross-exchange and inventory aribtrage strategy.} \label{fig:strategies}
    \begin{center}
     \includegraphics[clip, trim=2.5cm 6cm 2cm 2cm, width=1.00\textwidth]{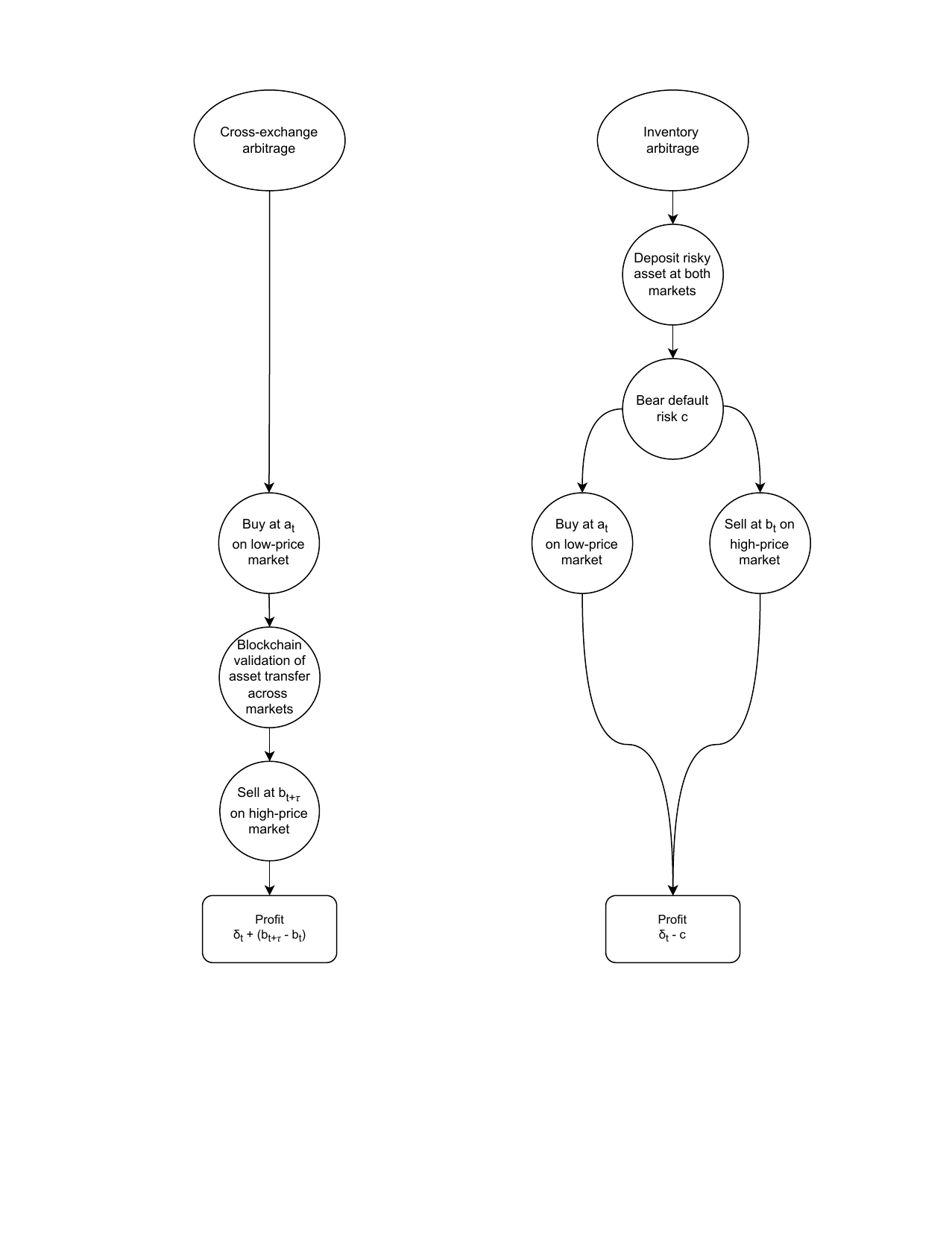}
    \end{center}
\end{figure}

We consider an economy containing a single blockchain-based asset traded on two centralized exchanges. We define an asset as blockchain-based if the securities' ownership is maintained on a shared database that can be updated without relying on trusted intermediaries or other third-party infrastructure. 
CEXes keep their order-matching systems off-chain, meaning they operate as escrows for their clients without recording transactions on the blockchain. 
Off-chain settlement has important benefits: First, no incentivizing payments to blockchain validators are needed, and second, transaction throughput can be scaled independently of the speed with which validators append new transactions to the blockchain. However, off-chain settlement is criticized as being vulnerable to massive breaches of security and unsafe storage of information, funds, and private keys.

We assume that the trading activity on both exchanges is exogenous and that agents can monitor the quotes of the asset across exchanges. 
Each exchange $k\in\{i,j\}$ continuously provides a log ask price $a^{k}_{t}$ and log bid price $b^{k}_{t}$ (with $b^{k}_{t} \leq a^{k}_{t}$) for one marginal unit of the asset at time $t$. 

Our sole agent is an arbitrageur who aims to exploit observed price differences across exchanges. 
She intends to buy a marginal unit of the asset on the exchange with a lower ask price and sell the same amount at the exchange with a higher bid price. Hence, in the case of frictionless trading, the arbitrageur exploits observed price differences when her profits are positive, i.e., whenever 
\begin{equation}
	\delta_t := \max\left\{b^i_t  - a^j_t, b^j_t  - a^i_t, 0\right\} > 0.
\end{equation} 
Cross-exchange price differences $\delta_t > 0$ can occur, for instance, as a response to private valuation shocks or asymmetric arrival of information \citep[e.g.,][]{Foucault.2017}. We assume throughout the paper that the arbitrageur cannot infer the low-price exchange before the cross-exchange price difference occurs. Upon the arrival of a positive price difference, we denote the resulting low-price exchange as $b\in\{i,j\}$ (for \emph{buy}) and the high-price exchange as $s$ (for \emph{sell}) such that $\delta_t = b^s_t  - a^b_t > 0$. 

\subsection{Cross-Exchange Arbitrage Strategies}

We distinguish between two major frictions when exploiting price differences of blockchain-based assets across CEXes: risk related to the latency in the settlement process and risk due to defaults of CEXes. 
To exploit price differences $\delta_t $ across two CEXes, the arbitrageur can allocate her arbitrage capital across two possible strategies, which we term \emph{cross-exchange arbitrage} and \emph{inventory arbitrage}. Figure~\ref{fig:strategies} illustrates the actions associated with the two arbitrage strategies.
A rational arbitrageur anticipates the implied costs of exploiting price differences and chooses the optimal arbitrage capital allocation. 

\paragraph{Cross-exchange arbitrage.}
Cross-exchange arbitrage requires storing funds of the numéraire at both CEXes. We keep the model parsimonious and assume that the numéraire can be stored without default risk under CEX custody.\footnote{Introducing such an additional source of uncertainty would not affect our results qualitatively because both cross-exchange and inventory arbitrage require deposits of the numéraire and are thus prone to the same additional source of risk. In our empirical analysis, we consider USD as safe numéraire deposits typically stored on CEX bank accounts outside the premises of vulnerable hacks.}
Deposits can be instantaneously exchanged for the blockchain-based asset on the corresponding CEXes. 
If buying on one exchange and selling on the other exchange implies a profit, the arbitrageur buys a marginal unit of the asset on the exchange with the lower ask price, transfers the asset to the exchange with a higher bid price, and sells the marginal unit as soon as the transfer is settled. 
Settlement in the context of blockchain-based assets is the validation of a transaction on the blockchain by validators, which indicates that the marginal unit of the asset has been deducted from a wallet under the control of the low-price CEX and has been credited to a wallet under the control of the desired high-price CEX. 
The absence of a clearing house for blockchain-based assets renders it impossible for the arbitrageur to dispose of her position before validation. 

We denote the settlement latency $\tau$ as the (known) waiting time until a transfer of the asset between exchanges is validated. 
The simplified setup with known $\tau$ illustrates the fundamental relationship between settlement latency and limits to arbitrage. In Section~\ref{sec:methodology}, we generalize the framework to random settlement latencies, transaction costs, optimal latency-reducing fee choice, and transaction quantities beyond marginal units. However, the main insights of the theoretical framework remain qualitatively the same.

Because the acquisition on the low-price exchange takes place at time $t$ and the transfer of the asset to the high-price exchange is settled at $t+\tau$, the arbitrageur faces the log bid price $b^{s}_{t+\tau}, s\in\{i, j\}$.
The profit of the arbitrageur's trading decision is thus at risk if the probability of losing money is non-zero. 
A risk-averse arbitrageur faces limits to (statistical) arbitrage in the spirit of \cite{Bondarenko.2003} whenever the associated risk exceeds the expected return. 
We model the random log price change on the high-price exchange from time $t$ to $t+\tau$ as a Brownian motion without drift, such that the log return of the cross-exchange arbitrage transaction is 
\begin{align}
	r_{(t:t+\tau)} := b^{s}_{t+\tau} - a_t^b = \underbrace{\delta_{t}}_{\genfrac{}{}{0pt}{}{\text{instantaneous}}{\text{return}}} + \underbrace{ b^{s}_{t+\tau} - b^{s}_t}_{\genfrac{}{}{0pt}{}{\text{exposure to }}{\text{price risk}}} = \delta_{t} + \sqrt{\tau}z,
\end{align}
where $z\sim N(0, \sigma^2)$ is normally distributed with volatility $\sigma$.
It follows that a risk-averse arbitrageur with mean-variance utility and coefficient of risk aversion $\gamma$ would only exploit price differences $\delta_t$ if her certainty equivalent ($CE$) is positive, i.e., 
\begin{align}\label{eq:simple_arbitrage_bounds}
CE = \mathbb{E}\left(	r_{(t:t+\tau)}\right) - \frac{\gamma}{2}\mathbb{V}\left(r_{(t:t+\tau)}\right) = \delta_t - \frac{\gamma}{2}\sigma^2\tau\geq 0 
\Longleftrightarrow	\delta_t \geq \frac{\gamma\sigma^2\tau}{2}.
\end{align}
Equation~\eqref{eq:simple_arbitrage_bounds} illustrates that settlement latency implies a risk-reward trade-off for risk-averse arbitrageurs. 
Whenever the observed price differences $\delta_t$ are positive, but CE is negative, the arbitrageur does not trade.\footnote{The risk aversion is associated with the arbitrageur's attitude towards the risk of a single trade. Theoretically, repeatedly exploiting price differences may lead to a vanishing variance of the arbitrageurs' aggregate returns, equivalent to a contraction of the relevant bounds. However, competition among arbitrageurs or information transmission across exchanges can imply a negative drift for the log price changes, which induces limits to arbitrage even for risk-neutral arbitrageurs \cite[e.g.,][]{Voigt.2020}.} 
 In this case, although the trade would be profitable under the possibility of instantaneous settlement, limits to arbitrage arise due to settlement latency. 
The arbitrageur requires higher expected returns if settlement latency $\tau$, the spot volatility $\sigma^2$, or risk aversion $\gamma$ increases. 

\paragraph{Inventory arbitrage.}
CEXes provide effective ways to circumvent settlement latency: To conduct inventory arbitrage, the arbitrageur deposits equal fractions of the available arbitrage capital as numéraire deposits \emph{and} collateral of the blockchain-based risky asset under the custody of each of the two exchanges. 
Note the important distinction to the numéraire when depositing the blockchain-based asset: such collaterals under the custody of exchanges can be subject to default risk.  
While capital-intensive, inventory arbitrage does not expose the arbitrageur to settlement latency. Instead, upon spotting a price difference $\delta_t > 0$, the arbitrageur acquires a marginal unit on the low-price CEX and immediately disposes of an offsetting amount of her inventory on the high-price CEX. 
As a result, her aggregate asset collateral holdings remain constant, but her aggregate numéraire deposits yield a return of $\delta_t$. 
In that sense, inventory arbitrage strategies effectively render price risk due to settlement latency obsolete. 
However, this benefit comes at the cost of storing the blockchain-based asset under the custody of the CEX without retaining corresponding private keys to control the holdings. 
As such, CEXes undermine the idea of decentralized finance without trusted intermediaries and reintroduce substantial counterparty risk. 
Similar to \cite{Biais.2019b} we model default risk as a fraction $c$ of deposits that get stolen while the arbitrageur waits for the arbitrage opportunity. Due to default risk, profits decrease proportional to wealth and reduce the returns on allocated capital.\footnote{For simplicity, we model default risks $c$ as a (stochastic) deposit loss upon arrival of the arbitrage opportunity living in $[0,1]$. Such an event may, for example, resemble the successful hack of the CEX in the spirit of \cite{Biais.2019b}. Without loss of generality, one can assume an underlying jump process such that a default occurs or does not occur with a (possibly stochastic) intensity.} 
After the arrival of the arbitrage opportunity, the inventory strategy thus delivers aggregate returns of $\delta_t - \mathbb{E}(c)$. We assume that the random variable $c$ exhibits known variance $\mathbb{V}(c) = \sigma_{c}^2$. 
Intuitively, higher perceived default risk in terms of higher expected losses due to default risk, $\mathbb{E}(c)$, corresponds to lower trust in CEXes. 

\subsection{Optimal Allocation and Limits to Arbitrage Capital}

For simplicity, we assume that the asset price dynamics are independent of default risk such that $Cov(z, c) = 0$.
The arbitrageur chooses $x_{\tau}$ and $x_{c}$ as the fractions of wealth allocated to the cross-exchange and inventory arbitrage strategies. The remaining fraction of wealth, $1 - x_\tau - x_c$, is invested into a risk-free outside option. As an outside option, the arbitrageur can decide to store her available capital in a private wallet, potentially paying interest $0 \leq r_f < \delta_t$ without any risk. 
Accordingly, her objective is 
\begin{align}
	\max_{x_\tau, x_c}CE\left(x_\tau, x_c\right) &= \left(x_\tau +x_c\right) \left(\delta_t - r_f\right) - x_c \mathbb{E}\left(c\right)- \frac{\gamma}{2} \left(x_\tau^2\sigma^2\tau + x_c^2\sigma^2_{c} \right).
\end{align}
The solution to the optimization problem yields the optimal allocation
\begin{align}
	x_\tau = \frac{\delta_t - r_f}{\gamma \sigma^2\tau} \text{ and } x_c = \frac{\delta_t - \mathbb{E}(c) - r_f}{\gamma \sigma^2_{c}}.
\end{align}
Hence, higher settlement latency ($\tau$) decreases the allocation into the cross-exchange strategy, $x_\tau$, and higher expected default risk $\mathbb{E}(c)$ or default risk uncertainty $\sigma^2_c$ reduces the allocation into the inventory strategy, $x_c$. 
Everything else equal, the allocations to the arbitrage strategies are not complements. Instead, higher risks reduce the aggregate capital allocated to the arbitrage strategies. 
The fraction of available arbitrage capital that is \emph{not} allocated into the risk-free outside option is
\begin{align}
	x_{\tau} + x_{c} = \frac{1}{\gamma\sigma_{c^2}}\left(\frac{\sigma_{c^2} + \sigma^2 \tau}{\sigma^2 \tau}\left(\delta_t - r_f\right) - \mathbb{E}\left(c\right)\right).
\end{align}
Hence, the arbitrageur allocates less capital to arbitrage activities when the risk aversion is higher, the price difference $\delta_t$ is smaller, the losses due to default risk are higher, or the settlement latency is higher. Taken together, the simple framework yields the following main insights from which we derive our testable predictions for the empirical analysis:
\begin{enumerate}
    \item \emph{Settlement latency implies limits to arbitrage.} Higher settlement latency $\tau$, spot volatility $\sigma$, and risk aversion $\gamma$ increase the required expected returns before exploiting price differences $\delta_t$ becomes incentive-compatible for the arbitrageur. Conversely, high cross-exchange price differences remain unexploited in periods of high settlement latency and asset return volatility.
    \item \emph{Trust into CEXes in the sense of low perceived default risk $\mathbb{E}\left(c\right)$ attracts more arbitrage capital and renders inventory strategies more attractive.} Settlement latency thus renders the exploitation of price differences across trustworthy CEXes less costly. As a result, the increased arbitrage capital renders price differences across exchanges with low default risk smaller.
    \item \emph{Order flow chases arbitrage opportunities until the resulting price pressure renders marginal trading costs too large.} A direct consequence of the decision to exploit price differences is the resulting order flow. Upon execution of the buy and sell transactions, arbitrageurs affect quoted prices on both CEXes.
\end{enumerate}

\section{Settlement Latency and Limits to Arbitrage}\label{sec:methodology}

As shown in a simplified framework in Section~\ref{sec:theory}, marginal costs of cross-exchange arbitrage increase with settlement latency, volatility, and risk aversion. 
In this section, we focus on the cross-exchange strategy and substantially relax the assumptions within the framework to reflect uncertainty in settlement latency and to allow for general utility functions. 
The theoretical framework is sufficiently general to reflect market frictions related to cryptocurrency trading such that we can empirically quantify costs for arbitrageurs due to settlement latency and the resulting economic consequences. 

More precisely, to empirically analyze the economic magnitudes and implications of settlement latency, we first relax the assumptions that volatility is constant over time and that the settlement latency is deterministic.  
Then, we introduce transaction costs and price impact. 
In our framework, we endogenize the optimal trade quantity and the optimal choice of settlement fees. 

First, we generalize the return dynamics compared to the simplistic framework from Section~\ref{sec:theory}. 
\begin{assumption}\label{assumption:price_process_general}
	For a given settlement latency $\tau$, we model the high-price exchange log price change $b^{s}_{t+\tau} - b^{s}_t$ as a Brownian motion with drift $\mu_t^s$ such that
	\begin{equation}\label{eq:return_process_with_drift}
		r^{b, s}_{(t:t+\tau)} = \delta^{b, s}_{t} + \tau\mu^s_t+ \int\limits_t ^{t+\tau} \sigma^s_{t} dW^s_k,
	\end{equation}
	where $\sigma^s_{t}$ denotes the spot volatility of the bid price process on the exchange $s$, and $W^s_k$ denotes a Wiener process. We assume that $\sigma^s_{t}$ is constant over the interval $[t, t+\tau]$.\footnote{Time-varying and stochastic volatility can be incorporated through a change of the timescale of the underlying Brownian motion. We provide the corresponding derivations in Appendix~\ref{sec:appendix_clock_change}. The time-variability of $\sigma^s_t$ and the presence of jumps would further increase the price risk the arbitrageur faces. In that sense, the bounds derived in this paper are conservative. }	
\end{assumption}
\noindent We now let the latency $\tau$ denote the \emph{random} waiting time until a transfer of the asset between CEXes is settled. Only weak assumptions regarding the stochastic nature of the settlement latency $\tau$ are required.
\begin{assumption}\label{assumption:latency_general}
	The settlement latency $\tau \in \mathbb{R}_+$  is a random variable equipped with a conditional probability distribution $\pi_t(\tau):=\pi\left(\tau|\mathcal{I}_t\right)$, where $\mathcal{I}_t$ denotes the set of available information at time $t$.
	We assume that the moment-generating function of $\pi_t(\tau)$, defined as $m_\tau\left(u\right) := \mathbb{E}_t\left(e^{u\tau}\right)$ for $u\in\mathbb{R}$, is finite on an interval around zero.
\end{assumption}
\noindent Assumptions~\ref{assumption:price_process_general} and \ref{assumption:latency_general} fully characterize the return distribution $\pi_t\left(r^{b, s}_{(t:t+\tau)}\right)$ through the interval of random length  from $t$ to $t+\tau$.
\begin{lemma}\label{lemma:characteristic_function}
	Under Assumptions~\ref{assumption:price_process_general} and \ref{assumption:latency_general}, $r^{b, s}_{(t:t+\tau)}$ exhibits the probability distribution
	\begin{equation}
		\pi_t\left(r^{b, s}_{(t:t+\tau)}\right) = \int_{\mathbb{R}_+}\pi_t\left(r^{b, s}_{(t:t+\tau)}\big|\tau \right)\pi_t\left(\tau\right)d\tau,
	\end{equation} 
	and corresponding characteristic function\footnote{The characteristic function fully describes the behavior and properties of a probability distribution. For a random variable $X$, $\varphi_{X}(u)$ is defined as $\varphi_{X}(u)=\mathbb{E}(e^{iuX})$, where $i$ is the imaginary unit and $u\in\mathbb{R}$ is the argument of the characteristic function.}
	\begin{equation}\label{eq:moments_general}
		\varphi_{r_{(t:t+\tau)}^{b, s}}\left(u\right) =e^{iu\delta^{b, s}_{t}}m_\tau\left(iu\mu^s_t-\frac{1}{2}u^2(\sigma^s_{t})^2\right).
	\end{equation}
\end{lemma}
\begin{proof}
	All proofs are provided in Appendix~\ref{sec:appendix_proofs}.
\end{proof}
\noindent To quantify the arbitrageur's assessment of risk, we equip her with a general utility function.
\begin{assumption}\label{assumption:utility_function}
	The arbitrageur has a utility function $U_\gamma(r)$
	with risk aversion coefficient $\gamma$, where $r$ is the log return implied by her trading decision. We assume $U'_\gamma(r) > 0$ and $U''_\gamma(r) < 0$.
\end{assumption}
\noindent While the main result in Section~\ref{sec:theory} are based on a particular utility function, we now characterize the optimal decision for the cross-exchange strategy under the general Assumptions~\ref{assumption:price_process_general} to~\ref{assumption:utility_function}. The arbitrageur maximizes her expected utility $\mathbb{E}_t\left(U_\gamma(r)\right)$, which we express in terms of the CE. 
We derive the CE of exploiting concurrent cross-exchange price differences in the following theorem.
\begin{theorem}\label{theorem:general}
	Under Assumptions~\ref{assumption:price_process_general} - \ref{assumption:utility_function}, the certainty equivalent resulting from the cross-exchange arbitrage trade is given by 
	\begin{align}\label{eq:theorem_general}
		CE =  &\delta_t^{b, s} + \mathbb{E}_t(\tau)\mu^s_t +\sum\limits_{k=2}^\infty\frac{U_\gamma^{(k)}\left(\delta_t^{b, s} + \mathbb{E}_t(\tau)\mu^s_t\right)}{k!U'_\gamma\left(\delta_t^{b, s} + \mathbb{E}_t(\tau)\mu^s_t\right)}\mathbb{E}_t\left(\left(r_{(t:t+\tau)}^{b, s} - \delta_t^{b, s} - \mathbb{E}_t(\tau)\mu^s_t\right)^k\right),
	\end{align}
	where $U_\gamma^{(k)}\left(r\right) := \frac{\partial^k}{\partial r^k}U_\gamma\left(r\right)$.
\end{theorem}
\begin{proof}
	All proofs are provided in Appendix~\ref{sec:appendix_proofs}.
\end{proof}

\noindent Theorem~\ref{theorem:general} allows us to compare the expected utility of executing the arbitrage trade versus staying idle (which yields a riskless return of zero). The arbitrageur is willing to exploit cross-exchange price differences if and only if the CE given by Equation \eqref{eq:theorem_general} is positive. 
\begin{definition}\label{def:ce_root}
	We define the {arbitrage bound} $d_t^s$ as the minimum price difference necessary such that the arbitrageur prefers to trade. Formally, $d_t^s$ is the maximum of zero and the unique root\footnote{By definition of the CE, we have $F(d)=U_\gamma^{-1}\left(\mathbb{E}_t\left(U_\gamma\left(d+\mu^s_t\tau+\int_{t}^{t+\tau}\sigma_t^sW_k^s\right)\right)\right)$. Since $U_\gamma'(r) > 0$, the expectation is increasing in $d$. Moreover, since $U_\gamma''(r) < 0$, the inverse $U_\gamma^{-1}(r) > 0$ is also strictly concave. Thus, $F(d)$ is strictly increasing and has a unique root.} of
	\begin{align}\label{definition:ce_root}
		F(d)= &d + \mathbb{E}_t(\tau)\mu^s_t + \sum\limits_{k=2}^\infty\frac{U_\gamma^{(k)}\left(d + \mathbb{E}_t(\tau)\mu^s_t\right)}{k!U_\gamma '\left(d + \mathbb{E}_t(\tau)\mu^s_t\right)}\mathbb{E}_t\left(\left(r_{(t:t+\tau)}^{b, s} - d - \mathbb{E}_t(\tau)\mu^s_t\right)^k\right).
	\end{align}
\end{definition}
\noindent 
Definition~\ref{def:ce_root} is a generalization of the arbitrage bounds derived in Equation~\eqref{eq:simple_arbitrage_bounds}. 
Below we follow \cite{Schneider.2015} and ignore the impact of higher order moments above the fourth degree of the Taylor representation in Equation~\eqref{definition:ce_root}.
Under the additional assumption of a power utility function, Lemma~\ref{lemma:optimal_strategy_iso} provides an analytical closed-form expression for $d_t^s$.
\begin{lemma}\label{lemma:optimal_strategy_iso}
	If, in addition to Assumptions \ref{assumption:price_process_general} and \ref{assumption:latency_general}, the arbitrageur has an isoelastic utility function $U_\gamma(r) := \frac{\left(1+r\right)^{1-\gamma}}{1-\gamma}$ with risk aversion coefficient $\gamma > 1$, the arbitrage bound for $\mu_t^s = 0$ is given by
	\begin{equation}
		d_t^s = \frac{1}{2}\sigma^s_t\sqrt{\gamma\mathbb{E}_t\left(\tau\right)+\sqrt{\gamma^2\mathbb{E}_t\left(\tau\right)^2+2\gamma(\gamma+1)(\gamma+2)\left(\mathbb{V}_t(\tau)+\mathbb{E}_t(\tau)^2\right)}}.
	\end{equation}
\end{lemma}
\begin{proof}
	All proofs are provided in Appendix~\ref{sec:appendix_proofs}.
\end{proof}
\noindent Hence, $d_t^s$ positively depends on (i)  the arbitrageur's risk aversion,~$\gamma$, (ii) the volatility on the high-price exchange, $\sigma^s_t$, (iii) the conditional expected waiting time until settlement, $\mathbb{E}_t\left(\tau\right)$, and $(iv)$ the conditional variance of the waiting time, $\mathbb{V}_t\left(\tau\right)$. 

\subsection{Transaction Costs}\label{sec:trcosts}
Most CEXes demand trading fees that agents pay upon the execution of an off-chain transaction. Market participants typically pay fees as a percentage of the trading volume. 
Similarly, broker-dealers usually charge markups for the execution of trades in over-the-counter exchanges. 
Moreover, exchanges typically exhibit limited supply through price-quantity schedules that agents are willing to trade, possibly leading to substantial price impacts for large trading quantities.
We make the following assumption to incorporate trading fees and liquidity effects into our framework.
\begin{assumption}\label{assumption:transaction_costs}
	Trading the quantity $q\geq0$ on exchange $i$ exhibits proportional transaction costs such that the average per unit bid and ask prices are
	\begin{align}
		B^{i}_t(q) &= B_t^{i}\left(1 - \rho^{i,B}(q)\right)\\
		A^{i}_t(q) &= A_t^{i}\left(1 + \rho^{i,A}(q)\right),
	\end{align}
	with $\rho^{i,B}(q)\geq 0$ and $\rho^{i,A}(q)\geq 0$, both monotonically increasing in $q$.
\end{assumption}
\noindent The presence of transaction costs changes the objective function of the arbitrageur, who focuses on maximizing returns net of transaction costs defined as  
\begin{align}\label{eq:transaction_cost_adjustment}
	\tilde r_{(t:t+\tau)}^{b, s} &= b^{s}_{t+\tau} - b^{s}_t + \delta_t^{b, s} - \log\left(\frac{1 + \rho^{b, A}(q)}{1 - \rho^{s, B}(q)}\right) \nonumber\\
	&= r_{(t:t+\tau)}^{b, s} - \log\left(\frac{1 + \rho^{b, A}(q)}{1 - \rho^{s, B}(q)}\right).
\end{align} 
Intuitively, Equation~\eqref{eq:transaction_cost_adjustment} shows that transaction costs increase the instantaneous return required to make the arbitrageur indifferent between trading and staying idle. The following lemma summarizes the arbitrageur's decision problem in the presence of transaction costs.
\begin{lemma}\label{lemma:trans_cost}
	Under assumptions \ref{assumption:price_process_general} - \ref{assumption:transaction_costs}, the arbitrageur prefers to trade a quantity $q>0$ over staying idle if \begin{equation}\label{eq:boundary_txcosts}
		\delta_t^{b,s} - \log\left(\frac{1 + \rho^{b, A}(q)}{1 - \rho^{s, B}(q)}\right) > d_t^s.
	\end{equation}	
\end{lemma}
\begin{proof}
	All proofs are provided in Appendix~\ref{sec:appendix_proofs}.	
\end{proof}
\noindent Most importantly in the context of blockchain-based assets, settlement fees play a pivotal role in the architecture of many consensus protocols. 
Validators typically receive a reward for confirming transactions which (at least partly) comprises fees that originators of transactions offer to provide validators incentives to prioritize the settlement of transactions that include a higher fee \citep[e.g., ][]{Easley.2017, Huberman.2021, Lehar.2022}. 

\begin{assumption}\label{assumption:settlement_fees}
	A settlement fee $f>0$ implies a latency distribution $\pi_{t}\left(\tau| f\right)$ that can be ordered in the sense that for $\tilde f > f$, $\pi_{t}\left(\tau| f\right)$ first-order stochastically dominates $\pi_{t}\left(\tau| \tilde f\right)$, i.e., $\mathbb{P}\left(\tau\leq x|\tilde{f}\right)>\mathbb{P}\left(\tau\leq x|f\right)$ for all $x\in\mathbb{R}_+$.
\end{assumption}

\noindent The ordering of latency distributions in Assumption~\ref{assumption:settlement_fees} implies a lower CE of trading for $\tilde f > f$.\footnote{We refer to \cite{Hadar.1969} and \cite{Levy.1992} for an explicit analysis of the relation between stochastic dominance and expected utility.} Denote by $d_{t}^s(f)$ the arbitrage bound associated with the latency distribution $\pi_{t}\left(\tau|f\right)$. Theorem~\ref{theorem:general} then implies that $d_{t}^s(f) > d_{t}^s (\tilde f)$, i.e., by paying a higher settlement fee, the arbitrageur can reduce the risk associated with settlement latency and becomes more likely to trade. For simplicity, we assume that $d_{t}^s(f)$ is differentiable such that Assumption~\ref{assumption:settlement_fees} implies $\frac{\partial}{\partial f}d_t^s(f)<0$. 

While settlement fees reduce the latency, they are costly for the arbitrageur. 
Since the arbitrageur does not hold inventory of the asset on the low-price exchange, she has to acquire the additional quantity $f$ to spend it in the settlement process. 
In line with the typical practical implementation of settlement fees, we assume that the arbitrageur has to pay the settlement fee in terms of the underlying asset. 
Given the transaction costs from above, the choice of $f$ thus also affects the optimal trading quantity $q$. 
The following lemma characterizes the arbitrageur's decision problem in the presence of transaction costs and settlement fees.

\begin{lemma}\label{lemma:settlement_fees}
	Under assumptions \ref{assumption:price_process_general} - \ref{assumption:settlement_fees}, the arbitrageur prefers to trade a quantity $q>0$ and pay a settlement fee $f>0$ over staying idle if \begin{equation}\label{eq:settlement_fees}
		\delta_t^{b,s} - \log\left(\frac{1 + \rho^{b, A}(q + f)}{1 - \rho^{s, B}(q)}\right) > d_t^s(f).
	\end{equation}	
\end{lemma}
\begin{proof}
	All proofs are provided in Appendix~\ref{sec:appendix_proofs}.
\end{proof}
\noindent Lemma~\ref{lemma:optimal_strategy_iso} shows that blockchain congestion directly affects the trading decisions of cross-exchange arbitrageurs. Consider a scenario where a blockchain-based platform generates excess demand for validation services, which manifests in a substantial increase in settlement fees. Marginal costs for cross-exchange arbitrageurs would increase, and thus violations of the law of one price may persist. Such \emph{spillovers} can be induced, for instance, by DEX arbitrage competition \citep[e.g.,][]{Capponi.2021} or Ransomware activity \citep[e.g.,][]{Sokolov.2021}. 

\subsection{Optimal Trading Quantity}
\noindent Trading a larger quantity might deliver higher total returns, but it comes at higher transaction costs on both the low-price and high-price exchange. Moreover, paying higher settlement fees leads to lower arbitrage bounds but at the cost of additional transaction costs on the low-price exchange. The arbitrageur's trading decision thus features a trade-off between $q$ and $f$ with endogenous arbitrage bounds. Formally, the arbitrageur aims to maximize total profits
\begin{equation}
	\underset{\left\{q,f\right\}\in\mathbb{R}_+^2}{\max} B_t^s\left(1-\rho^{s,B}(q)\right)q - A_t^b(1+\rho^{b,A}(q+f))(q+f) 
\end{equation}
subject to the constraint
\begin{equation}\label{eq:constraint}
	\delta_t^{b,s}- \log\left(\frac{1 + \rho^{b, A}(q+f)}{1 - \rho^{s, B}(q)}\right) \geq d_t^s(f).
\end{equation}
We characterize the arbitrageur's optimal choice of trading quantities and settlement fees in the following lemma.
\begin{lemma}\label{lemma:optimal_q_f}
	A total return maximizing arbitrageur only pays a settlement fee $f^*>0$ to trade a quantity $q^*>0$ if the following necessary conditions are met:
	\begin{align}
		\frac{1-\rho^{s,B}(q^*)}{q^*}&>\frac{\partial}{\partial q}\rho^{s,B}(q^*)  \label{eq:choice1} \\
		-\frac{\partial}{\partial f}d_t^{s}(f^*) &> \frac{\frac{\partial}{\partial q}\rho^{s,B}(q^*)}{1 + \rho^{s,B}(q^*)}. \label{eq:choice2}
	\end{align}	
	Otherwise, the arbitrageur optimally sets $f^*=0$. 
	Moreover, a total return maximizing arbitrageur chooses trading quantities $q^*>0$ and settlement fees $f^*\geq 0$ such that 
	\begin{align}\label{eq:optimal_q_f}
		\delta_t^{b,s} - \log\left(\frac{1 + \rho^{b, A}(q^* + f^*)}{1 - \rho^{s, B}(q^*)}\right) = d_t^s(f^*).
	\end{align}
\end{lemma}
\begin{proof}
	All proofs are provided in Appendix~\ref{sec:appendix_proofs}.
\end{proof}
\noindent The first part of the lemma provides conditions for the choice of the settlement fee. According to Equation~\eqref{eq:choice1}, the arbitrageur chooses a positive settlement fee as long as the marginal price impact for the trading quantity is below the average price impact. 
However, Equation~\eqref{eq:choice2} shows that reducing the arbitrage bound through a higher settlement fee must exceed the implied opportunity costs, i.e., the possible gain in selling a higher quantity. 
Consequently, the arbitrageur tends to pay a higher settlement fee if the high-price exchange is very liquid (keeping the marginal price impact low) and the settlement fee has a high impact on the arbitrage bound (i.e., reducing the latency and thus risk).
If any of these two conditions is violated, the arbitrageur optimally chooses not to pay any settlement fee but might still decide to trade.

The second part of the lemma states that the arbitrageur always chooses trading quantities and settlement fees such that the constraint in Equation~\eqref{eq:constraint} binds. If the constraint would not be binding, the arbitrageur could trade a larger quantity to increase her total returns at the expense of higher transaction costs. 

\section{Bitcoin Order Book and Network Data}\label{sec:data}

We gather novel and granular data on Bitcoin, the largest blockchain-based asset in terms of market valuation, to assess the economic relevance of settlement latency as friction for cross-exchange arbitrage. 

Testable implications of the theoretical framework are at least threefold: 
We show for general assumptions in Section~\ref{sec:methodology} that marginal costs for cross-exchange arbitrage increase with settlement latency, latency uncertainty, and the volatility of the blockchain-based asset.
In Section~\ref{sec:theory}, we showed lower perceived expected losses due to default risk render inventory strategies more attractive such that arbitrageurs decide to store funds under the custody of the CEX. 
As a result, order flow should chase cross-exchange price differences as long as these are below marginal costs for arbitrageurs. 

We first collect Bitcoin order book data to investigate price differences across a large sample of CEXes at high frequencies and to estimate the spot volatility of Bitcoin. Second, we enrich our data with Bitcoin blockchain network data to quantify settlement latency.  
We use the parametrization to provide empirical evidence for each of the main implications of settlement latency for blockchain-based trading highlighted above. 

\subsection{Bitcoin Order Book Data} 

We gather order book information from the application programming interfaces (APIs) of the 16 largest CEXes in terms of trading volume in January 2018 that feature BTC versus US Dollar. We retrieve all open buy and sell limit orders for the first 25 order book levels on a minute interval from January 1, 2018, to October 31, 2019.

Table~\ref{tab:descriptives_orderbook} gives the corresponding exchanges and provides summary statistics of our sample period's underlying order book data. We observe a strong heterogeneity of exchange-specific liquidity. 
For instance, whereas investors could have traded BTC versus US Dollar at \emph{Coinbase Pro} with an average spread of 0.45~USD, the average quoted spread at \emph{Gatecoin} has been about 337~USD. 
For most exchanges, however, the relative bid-ask spreads are comparable to those from equity exchanges such as NASDAQ or NYSE, where relative spreads range from 5 basis points (bp) for large firms to 38 bp for small firms \citep[e.g.,][]{Brogaard.2014}. 

\begin{table}[t!]
	\caption{Descriptive statistics of the order book sample}\label{tab:descriptives_orderbook}
	\centering
 \resizebox{\columnwidth}{!}{%
\begin{tabular}{lrrrrrrlllll}
  \toprule
  & Orderbooks & Spread (USD) & Spread (bp) & Taker Fee & With. Fee & Conf.  & Margin & Business & Region & Rating & USDT \\ 
  \midrule
Binance & 941,399 & 2.61 & 3.29 & 0.10 & 0.00100 & 2 & \cmark & \xmark & Other & A & \cmark \\ 
  Bitfinex & 938,703 & 0.62 & 0.74 & 0.20 & 0.08000 & 3 & \cmark & \cmark & Other & A & \xmark \\ 
  bitFlyer & 919,182 & 15.13 & 20.52 & 0.15 & 0.08000 &  & \cmark & \cmark & Japan & A & \xmark \\ 
  Bitstamp & 938,483 & 5.11 & 6.33 & 0.25 & 0.00000 & 3 & \xmark & \cmark & USA & A & \xmark \\ 
  Bittrex & 940,523 & 9.07 & 13.20 & 0.25 & 0.00000 & 2 & \xmark & \cmark & Other & B & \cmark \\ 
  CEX.IO & 936,378 & 11.73 & 15.07 & 0.25 & 0.00100 & 3 & \cmark & \cmark & UK & B & \xmark \\ 
  Gate & 907,874 & 81.24 & 90.92 & 0.20 & 0.00200 & 2 & \xmark & \xmark & USA &  & \cmark \\ 
  Gatecoin & 560,111 & 336.52 & 515.87 & 0.35 & 0.06000 & 6 & \xmark & \cmark & Other &  & \xmark \\ 
  Coinbase Pro & 941,539 & 0.45 & 0.54 & 0.30 & 0.00000 & 3 & \cmark & \cmark & USA & AA & \xmark \\ 
  Gemini & 912,944 & 2.57 & 3.40 & 1.00 & 0.00200 & 3 & \xmark & \cmark & USA & AA & \xmark \\ 
  HitBTC & 919,686 & 2.96 & 3.68 & 0.10 & 0.00085 & 2 & \xmark & \xmark & UK & B & \cmark \\ 
  Kraken & 936,970 & 2.63 & 3.24 & 0.26 & 0.00100 & 6 & \cmark & \cmark & USA & A & \xmark \\ 
  Liqui & 491,516 & 30.15 & 45.13 & 0.25 &  &  & \cmark & \xmark & Japan &  & \cmark \\ 
  Lykke & 918,768 & 44.04 & 57.95 & 0.00 & 0.05000 & 3 & \xmark & \xmark & Europe & D & \xmark \\ 
  Poloniex & 916,876 & 5.38 & 7.51 & 0.20 &  & 1 & \cmark & \xmark & Other & A & \cmark \\ 
  xBTCe & 887,289 & 13.34 & 17.87 & 0.25 & 0.00300 & 3 & \cmark & \xmark & Europe &  & \xmark \\ 
   \bottomrule
\end{tabular}}
	\subcaption*{\emph{Notes:} This table reports descriptive statistics of order book data used in our study. We gather high-frequency order book information of 16 exchanges by accessing the public application programming interfaces (APIs) every minute. \emph{Order Books} denotes the number of successfully retrieved order book snapshots between January 1, 2018, and October 31, 2019. \emph{Spread (USD)} is the average quoted spread in USD, \emph{Spread (bp)} is the average spread relative to the quoted best ask price (in basis points). \emph{Taker Fee} are the associated trading fees in percentage points relative to the trading volume. \emph{With. Fee} are the withdrawal fees in BTC. \emph{Conf.} refers to the number of blocks the exchange requires to consider incoming transactions valid. Empty cells indicate missing values. \emph{Margin} refers to the existence of BTC shorting instruments at the exchange. \emph{Business} indicates whether the exchange allows business accounts and hence access for institutional investors. \emph{Region} denotes the main regional area of business activity of the CEXes, based on the classification in \cite{Makarow.2018}. \emph{Rating} are November 2019 exchange ratings from data provider \emph{cryptocompare} (highest grade: AA, lowest grade: E). \emph{USDT} indicates if exchanges quote BTC against Tether or against USD directly. }
\end{table}

The exchanges also exhibit substantial heterogeneity in terms of trading-related characteristics.
Taker fees range from 0\% on \emph{Lykke} to 1\% on \emph{Gemini}. 
Other potential transaction costs are withdrawal fees that have to be paid upon the transfer of BTC from the exchange to any other exchange or private wallet address. 
Exchanges charge up to $0.003$~BTC for withdrawal requests, corresponding to roughly 30~USD in prices as of January 2018, irrespective of the withdrawn amount. 
Furthermore, exchanges have different requirements concerning the number of block confirmations before processing BTC deposits. 
For instance, \emph{Kraken} requires that incoming transactions must be included in at least six blocks. 
These requirements aim to reduce the possibility of an attack that aims at revoking previous transactions, i.e., a so-called 'double-spending attack'. 
In such a scenario, a potential attacker has to alter all blocks containing the corresponding transaction. The probability that an attacker catches up with the honest chain decreases exponentially with the number of blocks the attacker has to alter \citep{Nakamoto.2008}. 
As we discuss below, these requirements confront arbitrageurs with a mechanical increase in the settlement latency. 

We collect information about two exchange characteristics that might reduce marginal arbitrage costs. 
On the one hand, some exchanges offer margin trading instruments that allow traders to take short positions on BTC. However, such margin trading always comes at the cost of substantial collateral deposits, which the exchanges control. On the other hand, some exchanges allow businesses to open an account, which allows institutional investors, who might have lower risk aversion, to hold inventories and exploit price differences. However, holding inventories at exchanges is costly since it is associated with continuous exposure to exchange-specific default risk. 
We demonstrate in Section~\ref{sec:regressions} that the mere presence of margin trading instruments or access for institutional investors is not a sufficient condition to offset the impact of settlement latency.
Further, \cite{Makarow.2018} find that price deviations are much larger across than within countries, indicating the importance of capital controls for the movement of arbitrage capital. To control for such constraints, we group the exchanges in our sample into geographical regions according to the classifications in \cite{Makarow.2018}.  

Further, as one proxy for perceived expected losses due to the default risk of the CEX, we collect CEX ratings from data provider \emph{cryptocompare}. The ratings were created in November 2019 and are based on evaluating the business along a range of dimensions, including data provision, security, trade monitoring, and potential negative reports about the CEXes' activities. The ratings range from top (\emph{AA}) to dubious assessments (\emph{E}). Two exchanges in our sample were rated with the top grade \emph{AA}. 

Finally, some exchanges do not feature fiat currencies directly. However, they allow trading BTC against Tether, a token backed by one US Dollar for each token and trading close to par. In response to the results documented in \cite{Griffin.2020} substantial doubts on the backing of Tether by the US Dollar arose. We show in the Appendix that our empirical results remain qualitatively unaffected if we adjust for contemporaneous price deviations between the USD and USDT (Tether). 
 
\subsection{Bitcoin Network Data} \label{sec: bitcoin data}

To quantify the settlement latency for Bitcoin, we gather transaction-specific information from \href{https://www.blockchain.com/}{blockchain.com}, a popular provider of Bitcoin network data. We download all blocks validated between January 1, 2018, and October 31, 2019, and extract information about all validated transactions on the blockchain in this period. 
Each transaction contains a unique identifier, a timestamp of the initial announcement to the network, and, among other details, the fee (per byte) the transaction initiator offers validators to validate the transaction.\footnote{The fee per byte is more relevant than the total fee associated with a transaction as block sizes are limited in terms of bytes. In principle, a transaction can have multiple inputs and outputs, i.e., several addresses involved as senders or recipients of a transaction, which increases the number of bytes.}
 
Any transaction in the Bitcoin network, irrespective of its origin, must go through the so-called mempool, a collection of all unconfirmed transactions.
These transactions wait until they are picked up by validators and get validated. 
The difference between the timestamp of the first announcement of the block, which contains the broadcasted transaction, and the timestamp of the transaction's initial announcement constitutes the settlement latency. 
The size of the mempool thus reflects the number of transactions that wait for confirmation. We retrieve the minute-level number of transactions in the mempool from data provider blockchain.com. 
By design, the Bitcoin protocol restricts the number of transactions that can enter a single block. 
This restriction induces competition among the originators of transactions who can offer higher {settlement fees} to make it attractive for validators to include transactions in the next block.
Consequently, transactions with no or very low settlement fees may not attract validators and thus stay in the mempool until they eventually become validated. Relaxing this artificial supply constraint might reduce issues pertaining to settlement latency but at the cost of reduced network security \citep[see, e.g., ][]{Hinzen.2022}.

Validators bundle transactions that wait for verification and try to solve a computationally expensive problem that involves numerous attempts until the first validator finds the solution. 
By design of the Bitcoin protocol, validators successfully find a solution and append a block on average every 10 minutes (during our sample period, new blocks are announced to the network on average every 9.7 minutes).
Settlement latency should not be confused with the time it takes until a new block is mined. 
Even though the expected block validation time is 10 minutes, it is ex-ante uncertain when a transaction is included in a block for the first time. 
The number of outstanding transactions serves as a proxy for fluctuations in congestion of the Bitcoin network.
Whereas on average, 1,644 transactions have been included per block in our sample period, the average number of transactions in the mempool is above 10,000, with temporarily more than 41,000 transactions waiting for validation.
For any transaction, this induces uncertainty in the settlement latency.
We confirm below that the probability of being included in the next block decreases with the number of transactions that wait for settlement and increases with the settlement fee the investor is willing to pay. 

Table~\ref{tab:summary_raw_mempool} provides summary statistics of the recorded transactions. The average settlement fee per transaction is about 2~USD. 
The distribution of fees exhibits a strong positive skewness with a median of 0.28~USD. 
The average waiting time until the validation of a transaction is about 41 minutes, while the median is about 8.8 minutes. 

\begin{table}[t!]
	\caption{Descriptive statistics of transactions in the Bitcoin network}\label{tab:summary_raw_mempool}
	\centering
	\begin{adjustwidth}{-0.6cm}{}
		\begin{footnotesize}
\begin{tabular}{lrrrrrrr}
  \toprule
  & Mean & SD & 5 \% & 25 \% & Median & 75 \% & 95 \% \\ 
  \midrule
Fee per Byte (in Satoshi) & 47.41 & 183.08 & 1.21 & 5.00 & 14.06 & 45.52 & 200.25 \\ 
  Fee per Transaction (in USD) & 1.98 & 24.19 & 0.02 & 0.09 & 0.28 & 1.12 & 7.54 \\ 
  Latency (in Min) & 41.03 & 289.26 & 0.73 & 3.55 & 8.82 & 20.75 & 109.52 \\ 
  Mempool Size (in 1000) & 10.02 & 14.88 & 0.43 & 1.81 & 4.50 & 11.06 & 41.88 \\ 
  Transaction Size (in Bytes) & 507.28 & 2174.13 & 192.00 & 225.00 & 248.00 & 372.00 & 958.00 \\ 
   \bottomrule
\end{tabular}
		\end{footnotesize}
	\end{adjustwidth}
	\subcaption*{\emph{Notes:} This table reports descriptive statistics of the Bitcoin transaction data used in our study. The sample contains all transactions validated in the Bitcoin network from January 1, 2018, to October 31, 2019. Our sample comprises 139,704,737 transactions that are validated in 99,129 blocks. \emph{Fee per Byte} is the total fee per transaction divided by the size of the transaction in bytes in Satoshi, where 100,000,000 Satoshi are 1 Bitcoin. \emph{Fee per Transaction} is the total settlement fee per transaction (in USD). We approximate the US Dollar price by the average minute-level midquote across all exchanges in our sample. \emph{Latency} is the time until the transaction is either validated or leaves the mempool without validation (in minutes). \emph{Transaction Size} denotes the size of a transaction in bytes. \emph{Mempool Size} is the number of other transactions in the mempool when a transaction of our sample enters the mempool.}	
\end{table}

\subsection{Price Differences Across Exchanges} \label{sec: pricedifferences}

To provide systematic empirical evidence on the extent of violations of the law of one price, we compute the observed instantaneous cross-exchange price differences, adjusted for transaction costs, of all $120$ exchange pairs  (with the total number of exchanges $N=16$), defined as 
\begin{small}
	\begin{align}\label{eq:tildeDelta}
	\tilde \varDelta_t :=\begin{pmatrix}
	0 &\cdots & \tilde\delta_t^{N, 1}\\
	\vdots& \ddots & \vdots\\
	\tilde\delta_t^{1, N}& \cdots& 0 \\
	\end{pmatrix}=\begin{pmatrix}
	0 &\cdots & \tilde b_{t}^{1}\left(q_t^{N, 1}\right) - \tilde a_{t}^{N}\left(q_t^{N, 1}\right)\\
	\vdots& \ddots & \vdots\\
	\tilde b_{t}^{N}\left(q_t^{1, N}\right) - \tilde a_{t}^{1}\left(q_t^{1, N}\right)& \cdots& 0 \\
	\end{pmatrix},
	\end{align}
\end{small}
\noindent where $\tilde b_t^i(q_t^{i,j})$ is the transaction cost adjusted log bid price of $q_t^{i, j}$ units of the asset on exchange $i$ at time $t$ and $\tilde a_t^i(q_t^{i,j})$ is the transaction cost adjusted log ask price of $q_t^{i, j}$ units of the asset. 

In line with our definition in Section~\ref{sec:trcosts}, transaction costs are proportional to the trading quantity. 
We choose $q_t^{i, j}$ as the quantity that maximizes the resulting return for the exchange pair $i$ and $j$ given the prevailing order books at the time $t$, the taker fees of exchanges $i$ and $j$ and withdrawal fees of exchange~$j$.
Accordingly, we account for proportional exchange-specific taker fees (as reported in Table~\ref{tab:descriptives_orderbook}), which increase the average ask price and decrease the average bid price. 
We then use the resulting transaction cost-adjusted order book queues and apply a grid search algorithm to identify the trading quantity that maximizes the total return for each exchange pair. 
As a last step, we check if the resulting trading quantity exceeds the withdrawal fee that the low-price exchange charges for outgoing transactions (see Table~\ref{tab:descriptives_orderbook}). If the optimal trading quantity is below the withdrawal fee, we set the trading quantity to zero.
This data-driven approach thus mimics the strategy of an arbitrageur who aims to maximize profits by optimally accounting for the prevailing order book depth and other trading-related fees. 
As price differences can only be positive in one trading direction, we set negative price differences to zero as such scenarios (even without latency) do not correspond to arbitrage opportunities. 
The resulting matrix of price differences thus contains only non-negative values. 

Figure~\ref{fig:price_differences} shows the average price differences for each exchange pair.
The heatmap shows that some exchanges exhibit quotes that tend to deviate quite systematically from (nearly) all other exchanges. 
For instance, \emph{Bitfinex}, \emph{CEX.IO}, \emph{Gatecoin}, and \emph{HitBTC} quote on average higher bid prices than most other exchanges and thus exhibit large price differences when used as a high-price exchange. 
Conversely, other exchange pairs do not feature large average price differences. 
For instance, there are hardly any price differences whenever \emph{Coinbase Pro} or \emph{Kraken} serve as high-price exchanges. 
We report price differences similar to among others \cite{Kroeger.2017}, \cite{Choi.2018}, \cite{Makarow.2018}, and \cite{Borri.2021}. 
Consistently, we find that price differences have a tendency to be lower across exchange pairs active within the same geographical region and for exchange pairs quoting either exclusively BTC versus USD or exclusively BTC versus USDT. 

\begin{figure}[t!]
	\caption{Price differences between exchanges}\label{fig:price_differences}
	\centering	
	\includegraphics[width=0.9\textwidth]{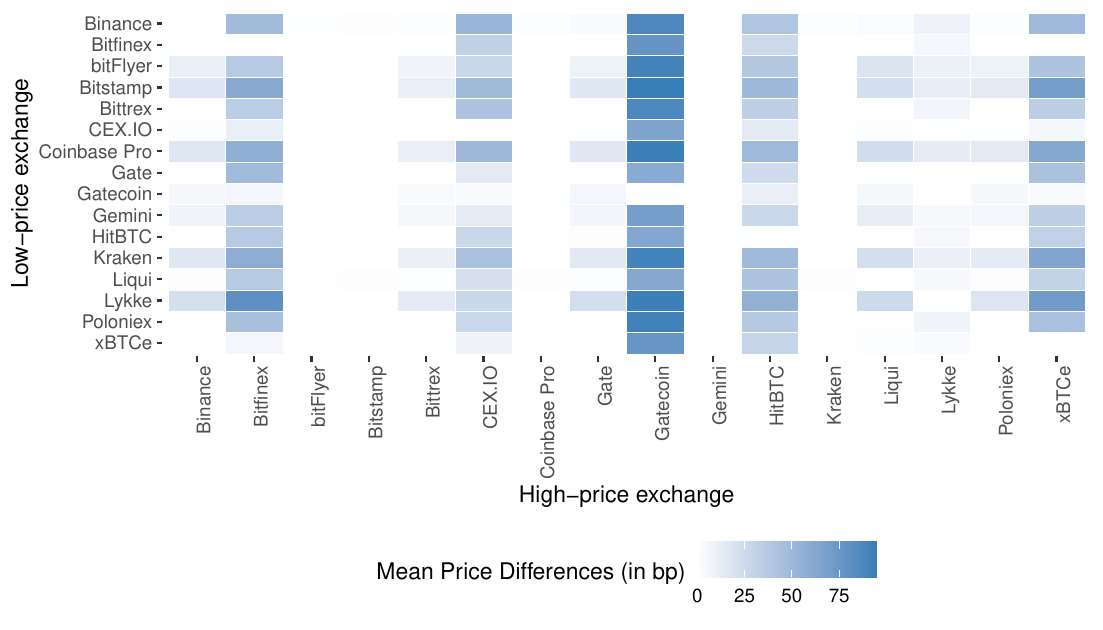}	
	\subcaption*{\emph{Notes:} The heatmap shows the average price differences, adjusted for transaction costs, $\tilde\delta_t^{b, s}$, across time for each exchange pair in our sample. Price differences are based on the minute-level transaction cost-adjusted bid and ask prices for each exchange according to Equation \eqref{eq:tildeDelta}. We account for exchange-specific taker fees according to Table~\ref{tab:descriptives_orderbook} and compute the quantity that maximizes each exchange pair's return using a grid search algorithm. The darker the color, the higher the average price difference through our sample period in the specific exchange pair. White or very light colors indicate that there are, on average, no or few price differences for a specific exchange pair.}	
\end{figure}

\section{Expected Costs Due to Settlement Latency} \label{sec: results}

Spot volatility $\sigma_t$, the expected settlement latency, and the variance of settlement latency are the central parameters of our theoretical framework in Section~\ref{sec:methodology}. We estimate these values based on our Bitcoin order book and network data.

\subsection{Spot Volatility}

To estimate the spot volatility, we follow the approach of \cite{Kristensen.2010}. For each exchange $s$ and minute $t$, we estimate $\left({\sigma}^s_{t}\right)^2$ by
\begin{equation}\label{eq:spotvola}
	\widehat{\left({\sigma}^s_{t}\right)}^2(h_T) = \sum\limits_{l=1}^\infty K\left(l - t, h_T\right)\left(b^{s}_{l} - b^{s}_{l-1}\right)^2,
\end{equation}
where $K\left(l - t, h_T\right)$ is a one-sided Gaussian kernel smoother with bandwidth $h_T$ and $b^s_{l}$ corresponds to the quoted bid price on the exchange $s$ at minute $l$. 
The choice of the bandwidth $h_T$ involves a trade-off between the variance and the bias of the estimator. 
Considering too many observations introduces a bias if the volatility is time-varying, whereas shrinking the estimation window through a lower bandwidth results in a higher variance of the estimator. 
\cite{Kristensen.2010} thus proposes to choose $h_T$ such that information on day $T-1$ is used for the estimation on day $T$. Formally, the bandwidth on any day of our sample is the result of minimizing the Integrated Squared Error of estimates on the previous day, i.e., 
\begin{align}
	h_T = \arg\min _{h>0}\sum\limits_{l=1}^{1440}\left[\left(b^{s}_{l} - b^s_{l-1}\right)^2 - \widehat{\left({\sigma}^s_{l}\right)}^2(h) \right]^2,
\end{align}
where $l$ refers to the minutes on day $T-1$ and $\widehat{\left({\sigma}^s_{l}\right)}^2(h)$ is the spot variance estimator for minute $l$ on day $T-1$ based on bandwidth $h$. 

For each exchange, we trim the distribution of all estimates at 1\% on both tails to eliminate outliers (e.g.,~due to  flickering quotes). 
Since the underlying asset is identical, the resulting estimates---as expected---do not differ substantially across exchanges. 
The average minute-level volatility across exchanges is about 0.09\%, which translates into daily volatility of about 3.4\%, significantly higher than the average daily volatility of the S\&P~500 index during the same period, which yields roughly 0.65\%.\footnote{We convert minute-level estimates to the daily level by multiplying it with the square root of the number of minutes on any given trading day, i.e.,  $\sqrt{1440}$.}

\subsection{Latency Prediction} \label{sec:latency}

We use all validated transactions to quantify the settlement latency of the Bitcoin blockchain. 
In line with \cite{Chiu.2018} and \cite{Easley.2017}, we expect that transaction fees and mempool congestion play an important role in the determination of the expected time until validation. 
Accordingly, we employ a Gamma regression, where the conditional probability density function of latency $\tau_i$ with rate parameter $\beta_i$ and shape parameter $\alpha_T$ is given by
\begin{align}\label{eq:duration_model_gamma}
	\pi(\tau_i|\theta_T)=\frac{\beta_i^{\alpha_T}}{\Gamma\left(\alpha_T\right)}\tau_i^{\alpha_T-1}e^{-\beta_i\tau_i},
\end{align}
where
\begin{align}
	\theta_T :=(\theta_T^{\beta}, \alpha_T)'\in\mathbb{R}^k \text{ and }\beta_i&=\exp(-x_{i}'\theta_T^{\beta}), \alpha_T > 0.
\end{align}
Here, $x_{i}\in\mathbb{R}^K$ includes an intercept and denotes (pre-determined) covariates driving $\tau_i$, $\theta_T^{\beta}\in~\mathbb{R}^K$ denotes the corresponding vector of parameters and $\Gamma\left(x\right):=\int_{\mathbb{R}_+}z^{x-1}e^{-z}dz$ is the Gamma function. 
The Gamma distribution collapses to an exponential distribution for $\alpha_T=1$. 
We estimate the parameter vector $\theta_T$ using all validated transactions on the day $T-1$ via maximum likelihood, both with and without covariates. 
In addition, we estimate an exponential model by fixing $\alpha_T=1$.
As covariates $x_i$, we include settlement fees and the (log) size of the mempool. 
The settlement fees enter as \emph{fees per byte} as the relevant metric for validators who face a restriction in terms of the maximum size of a block in bytes. 
Blockchain congestion, i.e., the number of transactions waiting for validation at the time when a transaction is announced, serves as a proxy for competition among transactions. 

In Table~\ref{tab:tab_exponential_model_fit}, we provide summary statistics of the estimated parameters. The numbers in the brackets denote the 5\% and 95\% quantiles of the time series of estimated parameters. 
The marginal effect of settlement fees is statistically significant and has the expected sign for nearly all days, i.e., higher fees predict a lower latency. 
The mempool size exhibits a positive impact on latencies through our sample period, i.e., congestion of the mempool decreases the probability of inclusion of a transaction in the next block \citep[see, e.g.,][]{Huberman.2021, Easley.2017}. 
A likelihood ratio test against a model without covariates indicates that the regressors are jointly highly significant. We, therefore, find clear evidence that the waiting time until a transaction enters the next block of the blockchain is predictable. 
We moreover find that the exponential distribution is rejected in favor of the more general Gamma distribution in nearly 93\% of all days. 

To predict the (conditional) moments of the latency distribution while avoiding any look-ahead bias, we use the estimated parameter $\hat{\theta}_{T}$ based on transactions from day $T-1$ to parameterize the latency distribution for every minute $t$ of day $T$. We provide further evidence for the predictability of settlement latency by computing the in-sample as well as out-of-sample root mean square prediction errors (MSPEs). In particular, for the in-sample MSPE, we use all transactions that feed into the estimation of $\hat{\theta}_{T}$ (i.e., all transactions validated on day $T-1$). The out-of-sample MSPE is based on predictions for all transactions validated on day $T$ using the estimated parameter vector $\hat{\theta}_{T}$. We find that the in-sample MSPE is, on average, smaller for the unrestricted model specifications and that the unrestricted models exhibit, on average, a lower out-of-sample MSPE compared to their restricted counterparts. As a consequence, we predict the latency using the unrestricted Gamma model. 

Accordingly, the conditional mean and variance of settlement latency $\tau$, induced by a transaction at minute $t$ on day $T$ with characteristics $x_t$, is given by
\begin{align}\label{eq:latency_moments}
\widehat{\mathbb{E}}_{t}\left(\tau\right) &= \hat\alpha_{T}\exp {(x'_t\hat{\theta}_{T}^\beta)}, \qquad\text{and }\qquad \widehat{\mathbb{V}}_{t}\left(\tau\right) = \hat\alpha_{T}\exp {(2x'_t\hat{\theta}_{T}^\beta)}, 
\end{align}
where $x_t$ consists of the mempool size and the fee an arbitrageur is willing to pay at time~$t$. While the mempool size is observable at any point in time, we use the optimal fee as derived in Lemma~\ref{lemma:optimal_q_f} as a proxy for the individually chosen settlement fees.
 
\begin{table}[t!]
	\centering
	\caption{Parameter estimates for the latency prediction models}\label{tab:tab_exponential_model_fit}
	\begin{footnotesize}
		\begin{tabularx}{\textwidth}{lcccc}
			\toprule
			                                     &   \multicolumn{2}{c}{Exponential}    &       \multicolumn{2}{c}{Gamma}       \\
			\cmidrule(lr){2-3}\cmidrule(lr){4-5} & W/o Covariates  &   W/ Covariates    &  W/o Covariates  &   W/ Covariates    \\
			\midrule
			Intercept                            &      3.31       &        1.41        &       3.86       &        1.19        \\
			                                     & {[2.510, 4.246]} &  {[-0.070, 3.675]}  & {[2.626, 5.250]}  &  {[0.013, 2.596]}  \\
			$\alpha$                             &                 &                    &       0.62       &        0.63        \\
			                                     &                 &                    & {[0.358, 0.902]} &   {[0.365, 0.900]}   \\
			Fee per Byte                         &                 &       -0.22        &                  &       -0.22        \\
			                                     &                 & {[-0.486, -0.031]} &                  & {[-0.501, -0.031]} \\
			Mempool Size                         &                 &        0.23        &                  &        0.31        \\
			                                     &                 & {[-0.043, 0.452]}  &                  &  {[0.059, 0.530]}   \\
			\midrule
			LR (Covariates)                      &      91.33      &                    &      74.59       &                    \\
			LR (Gamma vs. Exponential)           &      92.68      &                    &                  &                    \\
			\midrule
			MSPE (In-Sample)                     &      65.67      &       65.74        &      65.67       &       66.02        \\
			MSPE (Out-of-Sample)                 &      70.97      &       70.81        &      70.97       &       70.55        \\
			\bottomrule
		\end{tabularx}
	\end{footnotesize}
	\subcaption*{\emph{Notes: } This table reports summary statistics for the estimated parameters of the Gamma duration model  given by Equation \eqref{eq:duration_model_gamma}. \emph{Fee} denotes fee per byte and \emph{Mempool Size} refers to the number of unconfirmed transactions in the mempool. We estimate each model for each day in our sample, where  we consider all transactions validated on a particular day. We report the time-series averages of the estimated parameters. Values in brackets correspond to the 5\% and 95\% percent quantiles of the estimated parameters. \emph{LR (Covariates)} summarizes likelihood ratio tests of the corresponding unrestricted duration model \text{with} covariates against the restricted model \emph{without} covariates. \emph{LR (Gamma vs. Exponential)} summarizes  likelihood ratio tests of the Gamma duration model against the exponential specification. The reported values denote the percentage of days where the null hypothesis that the likelihood of the more general model equals the likelihood of the restricted model is rejected at the 95\% significance level. \emph{MSPE} refers to the mean squared prediction error for out-of-sample and in-sample tests, respectively.}
 \end{table}

We cannot reject the null hypothesis that the correlation between volatility and expected latency is significantly different from zero, which suggests that settlement latency constitutes a source of risk that is not captured by price fluctuations. 
In other words, periods of high spot volatility $\left({\sigma}^s_{t}\right)^2$ are not driven by high cross-exchange asset flows. Instead, settlement latency seems to be primarily driven by intraday fluctuation patterns of high network activity which are not necessarily related to cross-exchange arbitrage activity. For instance, \cite{Sokolov.2021} finds that network activity spikes during periods of ransomware attacks which increases network fees and thus, according to our theoretical framework, renders cross-exchange arbitrage activity more expensive. 

\subsection{Estimation of Arbitrage Bounds}

Having estimated spot volatilities $\left({\sigma}^s_{t}\right)^2$ as well as the first two moments of the settlement latency distribution, $\widehat{\mathbb{E}}_{t}\left(\tau\right)$ and $\widehat{\mathbb{V}}_{t}\left(\tau\right)$, we analyze the contribution of these components to the arbitrage bounds $d_t^s$. 
For that purpose, we estimate the arbitrage bounds according to our theoretical framework in Section~\ref{sec:methodology} and delineate $\hat d_t^s$ into individual components.

Based on the empirically relevant CRRA case of Lemma~\ref{lemma:optimal_strategy_iso}, the estimated arbitrage bounds $\hat{d}^s_{t}$ at minute $t$ are given by 
\begin{align}\label{equ:estimation_bounds}
	\hat{d}_t^s &= \frac{1}{2}\hat{\sigma}^s_t\sqrt{\gamma m_1+\sqrt{\gamma^2m_1^2+2\gamma(\gamma+1)(\gamma+2)m_2}},
\intertext{with}
	m_1 &= \widehat{\mathbb{E}}_{t}\left(\tau\right)  + \widehat{\mathbb{E}}_{t}\left(\tau_B\right)\cdot(B^s-1), \label{eq:c1} \\\label{eq:c2}
	m_2 &=\widehat{\mathbb{V}}_{t}\left(\tau\right) + \widehat{\mathbb{V}}_{t}\left(\tau_B\right)\cdot \left(B^s -1\right)^2  +  \left(\widehat{\mathbb{E}}_{t}\left(\tau_B\right)\cdot(B^s-1)+\widehat{\mathbb{E}}_{t}\left(\tau\right)\right)^2,
\end{align}
where $\hat{\sigma}^s_t$ denotes the square root of the estimated spot volatility on the high-price exchange, and $\widehat{\mathbb{E}}_{t}\left(\tau\right)$ and $\widehat{\mathbb{V}}_{t}\left(\tau\right)$ denote the estimated conditional mean and variance of the latency distribution, respectively.
Moreover, $B^s$ refers to the number of blocks the high-price exchange $s$ requires before considering incoming transactions valid (see Table \ref{tab:descriptives_orderbook}). 
This exchange-specific security requirement thus further increases the settlement latency beyond the waiting time until a transaction's validation in the first block.\footnote{\emph{bitFlyer} and \emph{Liqui} do not report a minimum number of confirmations. 
They rather use a discretionary system depending on the individual transaction and the state of the network. 
In this case, we assume the number of confirmations to be equal to the median across all exchanges that provide such information, which is 3.} 

We thus decompose the settlement latency into two components: the time it takes until a transaction is included in the blockchain (i.e., the first block), $\tau$, and the additional time $\tau_B$ until exchanges accept the transaction as de facto being immutable. 
While $\tau$ is partially under the control of the arbitrageur by her choice of the settlement fee, the validation time of subsequent blocks is exogenous. 
In fact, we do not find evidence against non-zero autocorrelation in waiting times and constant volatility in the block validation time. 
This evidence supports the notion that the validation times of blocks are partially under the control of the Bitcoin network and are internally impaired by the computational complexity of the underlying cryptographic problem. 
As a result, we can safely assume that the waiting times between subsequent blocks after the first one, which includes the current transaction, are independently and identically distributed. 
As validators append a new block on average every 9.7 minutes in our sample, we use this magnitude as the best-possible prediction of the time between two subsequent blocks, $\widehat{\mathbb{E}}_{t}\left(\tau_B\right)$. Accordingly, $\widehat{\mathbb{V}}_{t}\left(\tau_B\right)$ denotes the (sample) variance of the time between two consecutive blocks.

We fix the coefficient of risk aversion to $\gamma=2$ and estimate $\hat{d}_t^s$ for each exchange on a minute level.\footnote{Our estimation follows \cite{Conine.2017}, who estimate an average coefficient of relative risk aversion of about 2 over an extensive sample period.} Table~\ref{tab:efficiency_bounds} gives summary statistics of the resulting time series of arbitrage bounds due to settlement latency. We observe that the estimated bounds range, on average, between 91 bp and 197 bp. We acknowledge that the choice of the risk aversion coefficient determines the level of the estimated arbitrage bound $\hat{d}_t^s$. 
In that sense, interpreting the magnitude of the arbitrage bound itself is challenging. 
However, instead of only interpreting $\hat{d}_t^s$ in relation to observed price differences, we can focus on the relative importance of the core components that determine $\hat{d}_t^s$ in Equation~\eqref{equ:estimation_bounds}.

\begin{table}[t!]
	\centering
	\caption{Summary of exchange-specific arbitrage bounds}\label{tab:efficiency_bounds}	
	\begin{footnotesize}
		\begin{tabularx}{\textwidth}{lrrrrrrrrr}
  \toprule
  & Mean & SD & 5\% & 25\% & Median & 75\% & 95\% & Security & Uncertainty \\ 
  \midrule
Binance & 114.75 & 318.76 & 24.35 & 42.10 & 68.92 & 125.59 & 320.28 & 13.54 & 41.53 \\ 
  Bitfinex & 117.22 & 299.25 & 18.89 & 42.47 & 73.26 & 136.19 & 324.19 & 23.98 & 40.85 \\ 
  bitFlyer & 130.85 & 317.68 & 33.02 & 57.07 & 86.62 & 145.18 & 333.88 & 24.09 & 40.72 \\ 
  Bitstamp & 126.34 & 294.72 & 28.45 & 50.53 & 80.46 & 145.61 & 341.72 & 23.69 & 40.79 \\ 
  Bittrex & 129.03 & 277.80 & 30.94 & 57.25 & 89.41 & 143.51 & 333.37 & 14.32 & 41.63 \\ 
  CEX.IO & 120.84 & 286.39 & 29.46 & 52.72 & 81.69 & 136.05 & 305.50 & 24.44 & 40.60 \\ 
  Gate & 101.50 & 277.20 & 24.12 & 43.81 & 68.78 & 117.27 & 260.03 & 14.04 & 41.48 \\ 
  Gatecoin & 196.89 & 219.90 & 2.62 & 46.70 & 118.29 & 274.82 & 638.77 & 45.95 & 40.26 \\ 
  Coinbase Pro & 114.84 & 305.25 & 17.89 & 40.75 & 71.77 & 132.79 & 318.48 & 24.44 & 40.68 \\ 
  Gemini & 115.36 & 343.30 & 21.07 & 43.27 & 72.42 & 130.54 & 309.53 & 24.44 & 40.77 \\ 
  HitBTC & 101.22 & 287.97 & 19.10 & 37.64 & 62.72 & 112.79 & 273.14 & 14.14 & 41.36 \\ 
  Kraken & 135.07 & 271.66 & 25.37 & 54.09 & 91.53 & 164.15 & 357.11 & 41.86 & 40.50 \\ 
  Liqui & 90.79 & 60.20 & 23.51 & 49.96 & 77.40 & 115.62 & 201.88 & 28.97 & 39.98 \\ 
  Lykke & 133.43 & 379.31 & 18.58 & 44.51 & 80.57 & 150.73 & 381.17 & 25.21 & 40.61 \\ 
  Poloniex & 94.69 & 264.09 & 18.49 & 33.32 & 55.53 & 104.34 & 260.68 & 0.00 & 45.13 \\ 
  xBTCe & 106.16 & 246.56 & 19.90 & 40.74 & 70.58 & 131.44 & 281.96 & 24.15 & 40.78 \\ 
   \bottomrule
\end{tabularx}

	\end{footnotesize}
	\subcaption*{\emph{Notes:} This table provides descriptive statistics of estimated arbitrage bounds for each high-price exchange. We compute arbitrage bounds for a CRRA utility function with risk aversion coefficient $\gamma = 2$. We estimate the bounds using the spot volatility estimator of \cite{Kristensen.2010} and out-of-sample predictions of the conditional moments of the latency based on a Gamma duration model. We report all values in basis points (except otherwise noted). \emph{Security} gives the (percentage) contribution of the required number of confirmations to the median arbitrage bound. \emph{Uncertainty} corresponds to the (percentage) contribution of the uncertainty in latency to the median arbitrage bound.}
\end{table}

While the conditional moments of the latency distribution affect the time series variation of the arbitrage bounds, the cross-sectional variation is driven by the exchange-specific spot volatilities and the required number of confirmations, $B^s$. For instance, \emph{Gatecoin} and \emph{Kraken} require $B^s=6$ confirmations and produce, on average, the highest bounds, while \emph{Poloniex} requires only $B^s=1$ confirmation yielding the smallest median bound.
To quantify the effect of the exchange-specific security component $B^s$, we decompose the arbitrage bounds into the component resulting from the latency until a transaction is included in a block for the first time, $\tau$, and the component resulting from the waiting time until a transaction fulfills exchange-specific security requirements, $(B^s-1)\tau_B$. The second to last column in Table~\ref{tab:efficiency_bounds} gives the increase in the median arbitrage bound when we take the exchange-specific number of confirmations into account. The values correspond to the (percentage) difference between the median arbitrage bound as of Equation~\eqref{equ:estimation_bounds} and the respective bounds based on the assumption $B^s = 1$ for all exchanges. We observe that the impact of exchange-specific security components on arbitrage bounds is substantial and accounts, on average, for 23\% of the bounds.

Moreover, our theoretical framework allows us to directly analyze the relevance of the latency \emph{variance}. 
As the uncertainty of the arbitrageurs' returns increases with the variance of the settlement latency, we can compare the estimated arbitrage bounds to the (hypothetical) case of a \emph{deterministic} latency. 
The last column in Table~\ref{tab:efficiency_bounds} reports the percentage increase in arbitrage bounds when adjusting for the randomness in latency. 
The values correspond to the percentage difference between the median arbitrage bound and bounds based on the assumption  $\mathbb{V}_{t}(\tau) =\mathbb{V}_{t}(\tau_B) = 0$. 
We find that the impact of the latency volatility is substantial and accounts, on average, for 41\% of the arbitrage bounds. 

\section{Settlement Latency and Cross-Exchange Activity} \label{sec:regressions}

In this section, we empirically investigate the major predictions of our analysis in depth: 
Based on our findings in Section~\ref{sec:theory}, we hypothesize that (i) settlement latency is a costly friction for cross-exchange arbitrageurs, (ii) that mitigating CEX default risk reduces marginal cross-exchange arbitrage costs, and (iii) order flow chases arbitrage opportunities. The general representation of arbitrage costs due to settlement latency in Section~\ref{sec:methodology} forms the underlying base to account for transaction costs and the role of settlement fees in the validation process. More precisely, we reported, in line with existing literature, substantial price differences across CEXes and derived estimates of the relevant determinants of price risk related to settlement latency in Section~\ref{sec:data}. Next, we analyze the link between settlement latency and violations of the law of one price to document the economic relevance of settlement latency as a friction in blockchain-based trading. 

\subsection{Settlement Latency and Price Differences}

First, we investigate the relationship between the observed price difference and expected settlement latency, latency volatility, and spot volatility. 
The theoretical analysis in Section~\ref{sec:theory} yields that marginal costs for arbitrageurs increase with settlement latency. In turn, periods of high price risk for arbitrageurs due to settlement latency are consistent with persistent large observed price differences.

Table~\ref{tab:motivating_regression} provides the estimation results of linear regressions of hourly averages of cross-exchange price differences on exchange pair fixed effects (FE) and various regressors. 
In columns (1) and (2), we include the average estimated exchange-specific arbitrage bound or, alternatively, its underlying components, i.e., the average hourly spot volatility on the high-price exchange, the hourly median, and the variance of realized waiting times of transactions entering the mempool until being included in a block for the first time. 
Consistent with our theoretical framework, we find a statistically significant positive relationship between cross-exchange price differences and the components of the arbitrage bounds. 
Substituting the estimated arbitrage bounds by their components confirms that large price differences are consistent with periods of high price risk due to settlement latency. 

In columns (3) and (4), we interact the estimated arbitrage bounds with high-price exchange-specific dummy variables indicating whether the exchange offers margin trading instruments (\emph{Margin}) and access for institutional traders (\emph{Business Accounts}). 
We find that exchanges with margin trading are less sensitive to settlement latency but still yield a significant relationship between price differences and settlement latency. 
The finding indicates that the costs of margin trading for investors tend to exceed the risk-adjusted latency-implied price risk, presumably due to substantial margin requirements by CEXes, which, in turn, implies a higher default risk. 
Similarly, exchanges that feature access for institutional traders are less sensitive to arbitrage bounds, consistent with the notion that large institutions are more likely to exhibit a lower risk aversion than individual arbitrageurs.

In the last two columns of Table~\ref{tab:motivating_regression}, we control for the number of Bitcoins under the custody of a CEX as a proxy for trust. 
In Section~\ref{sec:theory}, we show that if investors perceive the default risk associated with exchanges as low, they should be willing to store more of their holdings directly under the custody of CEXes. 
Thus, the observed number of Bitcoins serves as a proxy for the perceived default risk of the CEX. 
We extract the number of Bitcoins under the control of wallets that the data provider \emph{glassnode} associates with each high-price exchange directly from the Bitcoin blockchain. 

The time series of aggregate inventories, measured as the number of Bitcoins, at CEXes exhibits an annualized aggregate growth of 13.4\%. 
Given the closing prices on the last day of our sample, USD 12.4 Billion worth of Bitcoins were under the custody of CEXes. 
The data thus indicate that the perceived default risk of CEXes decreased during our sample period. 
Note, however, that whereas some exchanges increased their inventory by large amounts (e.g., 400\% at Coinbase, 141\% at Binance, and 32\% at Bitstamp), some other exchanges face net inventory outflows in our sample period. 
For cross-exchange trading in such cases, settlement latency thus became more relevant due to the higher implied costs of inventory arbitrage strategies.
Consistent with our testable predictions, we find in Table~\ref{tab:motivating_regression} that cross-exchange price differences tend to be narrower between CEXes with more funds under custody. 

As a robustness check, column (1) of Table~\ref{tab:robust_regions} in Appendix~\ref{sec:appendix_robustness} provides an alternative specification to investigate the relevance of trust in CEXes as a proxy for the marginal costs of inventory arbitrage strategies. While we identify inventory under the control of a CEX as a suitable proxy for trust in Section~\ref{sec:theory}, exchange ratings may provide an alternative measure for the reliability and, thus, perceived default risks of CEXes. We find that arbitrage bounds interacted with a dummy for CEXes that received the best rating (AA) from data provider \emph{cryptocompare} in November 2019, yields results that are qualitatively similar to our main regression specification: Exchanges that are perceived as more trustworthy, i.e., that have higher ratings, exhibit lower price differences, which is in line with the notion that inventory arbitrage is less costly and thus attracts more arbitrage capital.

\begin{table}[ht]
	\caption{Price differences and sources of price risk}\label{tab:motivating_regression}
\centering
 \resizebox{\columnwidth}{!}{%

\begin{tabular}{lllllll}

\toprule
  \emph{Dependent Variable:} & \multicolumn{6}{c}{Price Differences (in \%)}\\
                               & (1)           & (2)            & (3)           & (4)            & (5)             & (6)\\  
   \midrule
   Arbitrage Bound (in \%)             & 4.993$^{***}$ &                & 5.178$^{***}$ & 6.078$^{***}$  & 3.193$^{***}$   &   \\   
                                       & (23.67)       &                & (21.11)       & (18.74)        & (14.50)         &   \\   
   Spot Volatility (in \%)             &               & 1.160$^{***}$  &               &                &                 & 0.7281$^{***}$\\   
                                       &               & (22.51)        &               &                &                 & (13.24)\\   
   Latency Median                      &               & 0.7004$^{***}$ &               &                &                 & 0.9760$^{***}$\\   
                                       &               & (5.212)        &               &                &                 & (6.833)\\   
   Latency Variance                    &               & 0.4306$^{***}$ &               &                &                 & 0.4371$^{***}$\\   
                                       &               & (6.640)        &               &                &                 & (6.348)\\   
   Arbitrage Bound $\times$ Margin Trading    &               &                & -0.3646       &                &                 &   \\   
                                       &               &                & (-0.9056)     &                &                 &   \\   
   Arbitrage Bound $\times$ Business Accounts &               &                &               & -1.847$^{***}$ &                 &   \\   
                                       &               &                &               & (-4.379)       &                 &   \\   
   Inventory                           &               &                &               &                & -0.1447$^{***}$ & -0.1505$^{***}$\\   
                                       &               &                &               &                & (-3.501)        & (-3.644)\\   
   \midrule
   Exchange Pair FE      & Yes           & Yes            & Yes           & Yes            & Yes             & Yes\\  
   High-Price Exchange FE                          & Yes           & Yes            & Yes           & Yes            &  No               & No \\  
   Observations                        & 1,416,752     & 1,416,752      & 1,416,752     & 1,416,752      & 1,416,752       & 1,416,752\\  
   \bottomrule
\end{tabular}

 }
\subcaption*{\emph{Notes:} This table provides OLS estimates based on a regression of hourly average exchange pair price differences and the main components of price risk due to settlement latency. \emph{Price Difference} is the average hourly exchange pair price difference (in percent). \emph{Spot Volatility} is the average hourly high-price exchange spot volatility estimate based on one-sided Gaussian kernel estimates \citep{Kristensen.2010}. \emph{Latency} denotes the (log) average hourly median (variance) of the waiting time (in minutes) of transactions entering the Bitcoin mempool. \emph{Arbitrage Bound} corresponds to the average hourly estimated arbitrage bound on the corresponding high-price exchange. \emph{Margin} is a dummy variable that indicates the availability of margin trading instruments on the high-price exchange, and \emph{Business Accounts} indicates whether the high-price exchange offers access for institutional investors. \emph{Inventory} is the average number of Bitcoins controlled by all wallets associated with the high-price exchange at hour $t-1$. We report $t$-statistics based on standard errors following \cite{Newey.1987} using 24 lags in parentheses. $^{***}, ^{**}$, and $^{*}$ indicate statistical significance on the $1\%, 5\%$ and $10\%$ levels (two-tailed), respectively.}
\end{table}

A well-established friction for cross-exchange cryptocurrency arbitrage are capital controls which may hamper arbitrage activity due to regulatory constraints. These constraints may prevent arbitrageurs from actively exploiting price differences on exchanges that are active in different countries. \cite{Makarow.2018} indeed document that price differences are smaller across exchange pairs that are active within the same geographic area. To shed further light on the relevance of settlement latency for the marginal costs of arbitrageurs, we decompose price differences into regional components that are based only on within-region cross-exchange pairs.
Table~\ref{tab:robust_regions} in the Appendix provides the regression results in columns (2) - (4). 
Even when we focus on price differences across exchange pairs that are exclusively active in the same region (only \emph{USA} (2), only \emph{Europe} (3), or \emph{all} exchange pairs where both exchanges are active within the same region (4)), price differences tend to be higher when settlement latency renders arbitrage more costly. Settlement latency thus plays an important role in understanding the dynamics of price differences, even when capital controls do not impose any costs for arbitrageurs. The magnitudes of the estimated coefficients, however,  indicate that price differences are less sensitive to arbitrage bounds within regions, confirming the relevance of capital controls to understanding the cross-section of price differences. 

In a similar vein, we find that settlement latency remains relevant beyond local fiat stability. While some exchanges do feature trading of BTC versus USD, others instead quote BTC versus USDT, a stablecoin which is, in theory, pecked to the US Dollar. Among others, \cite{Griffin.2020} documented unbacked digital money inflating cryptocurrency prices such that substantial doubts on the backing of Tether by the US Dollar arose. In columns~(5) - (7) in Table~\ref{tab:robust_regions} in the Appendix, we show that even if we restrict the sample in our regression setup to exchange pair price differences within CEXes that feature trading of BTC versus USDT, settlement latency remains a relevant channel which affects price differences. The same holds for the subset that just features CEXes with trading of BTC versus USD. Thus, the channel established in this paper seems robust to the presence of other well-documented frictions within cryptocurrency trading. 

\subsection{Settlement Latency and Cross-Exchange Flows}

In the last step of our analysis, we exploit the estimated arbitrage bounds to shed light on the relationship between cross-exchange price differences and transfers of assets between CEXes. 
Our theoretical analysis implies that during periods of large price differences, i.e., in periods where price differences likely exceed the arbitrage bounds, transfers of funds between CEXes should increase. 
 
We, therefore, extend our data by cross-exchange asset flows. 
Since exchanges are reluctant to provide the identity of their customers, it is virtually impossible to identify actual transactions by arbitrageurs. 
However, we take the overall transfer of assets between two different exchanges as a measure of the trading activity of cross-exchange arbitrageurs. 
For each exchange, we thus collect a list of addresses likely under the control of the exchanges in our sample.\footnote{We thank Sergey Ivliev for his tremendous support on this front.} 

Bitcoin transactions are pseudonymous in the sense that each transaction publicly reveals all addresses associated with the transaction. Still, it is hard to map these addresses to their respective physical or legal owners. 
Exchanges typically control many addresses to keep track of individual users' assets. 
However, algorithms are available that link addresses to certain exchanges \citep[e.g.,][]{Meiklejohn.2013, Foley.2018}. 
Usually, the matching procedure relies on observing an address advertised to belong to an exchange or actively sending small amounts of Bitcoin to exchanges. 
We gather 62.6 million unique exchange addresses, which allow us to identify 3.9 million cross-exchange transactions with an average daily volume of USD 72 million in our sample period.\footnote{We compute the average daily volume by extracting the hourly sum of net flows (inflows to an exchange minus the outflows in BTC) and multiplying it by the hourly average midquote across all exchanges.} 

Table~\ref{tab:flows} gives the estimates of a two-stage least squares regression of hourly cross-exchange net flows on hourly averaged cross-exchange price differences as well as exchange pair effects. 
The dependent variable is the sum of cross-exchange net flows into the corresponding high-price exchange.

We have to take into account that cross-exchange net flows and price differences are jointly determined, giving rise to a simultaneity problem. 
On the one hand, arbitrage activity is expected to increase with higher price differences (in excess of arbitrage bounds). 
On the other hand, price differences should decrease in response to arbitrage trades as arbitrageurs enforce adjustments towards the law of one price. 
We, therefore, instrument the price differences by the estimated arbitrage bounds (columns (2) and (4)) and, alternatively, by their respective components, i.e., the spot volatility, median settlement latency, and variance of realized latencies (columns (1) and (3)). 
These variables satisfy the two necessary conditions for the validity of an instrument. 
First, we find a positive correlation between price differences and arbitrage bounds after controlling for other exogenous variables (see Table~\ref{tab:motivating_regression}). 
Second, the only role arbitrage bounds play in influencing cross-exchange flows is through their effect on endogenous price differences.

\begin{table}[t!]
	\centering
	\caption{Cross-exchange flows and arbitrage opportunities}\label{tab:flows}
	\begin{footnotesize}
		\begin{tabularx}{\textwidth}{l*{4}{Y}}
    \toprule
   \emph{Dependent Variable} & \multicolumn{2}{c}{Exchange net inflows (in 100k USD)} & \multicolumn{2}{c}{Exchange net inflows (in BTC)} \\ \cmidrule(lr){2-3} \cmidrule(lr){4-5}
                         & (1)            & (2)            & (3)            & (4)\\  
   \midrule
   Price Differences (in \%)   & 0.0062$^{***}$ & 0.0058$^{***}$ & 0.0644$^{***}$ & 0.0551$^{***}$\\   
                               & (7.807)        & (6.091)        & (7.428)        & (5.717)\\   
   \midrule
   Exchange Pair FE & Yes            & Yes            & Yes            & Yes\\  
   High-Price Exchange FE     & Yes            & Yes            & Yes            & Yes\\  
   Observations                & 1,416,752      & 1,416,752      & 1,416,752      & 1,416,752\\  
   \bottomrule
		\end{tabularx}
	\end{footnotesize}	
  	\subcaption*{\emph{Notes:} This table provides the estimated marginal effects based on a two-stage least square regression of cross-exchange asset flows on price differences. \emph{Net inflows} is the difference between the average hourly inflows (in BTC) to the corresponding high-price exchange and the average hourly outflow from the high-price exchange per exchange pair. Inflows denoted in USD are based on the concurrent average high-price exchange midquote. 
	\emph{Price Differences} are the fitted values of the regression outlined in Table ~\ref{tab:motivating_regression} and denote the average hourly exchange pair price difference. 
    In columns (1) and (3), we instrument price differences with all components of arbitrage bounds. Columns (2) and (4) correspond to the estimation results, where we directly use the estimated arbitrage bounds as an instrument. We report $t$-statistics based on standard errors following \cite{Newey.1987} using 24 lags in parentheses. $^{***}, ^{**}$, and $^{*}$ indicate statistical significance at the $1\%, 5\%$ and $10\%$ levels (two-tailed), respectively.}
\end{table}
Throughout all specifications, we find a significant positive relationship between cross-exchange net flows into an exchange and (instrumented) price differences. Hence, the regression results indicate that cross-exchange flows increase in response to larger price differences triggered by larger arbitrage bounds. 
\section{Conclusions}\label{sec:conclusion}

Many market participants believe that blockchain technologies have the potential to radically transform the transfer of financial assets. 
Replacing trusted intermediaries and central clearing parties with blockchain technologies may increase efficiency and security \emph{and} may lower transaction costs. However, new frictions emerge for blockchain-based assets as the potential merits come at the cost of latency in the settlement process. 
The inability to perform instantaneous cross-exchange transactions implies limits to arbitrage as market participants cannot react sufficiently fast to potential violations of the law of one price.

We show that settlement latency implies limits to arbitrage as it is not incentive-compatible for risk-averse arbitrageurs to exploit cross-exchange price differences in periods with high volatility and long validation times. We formally derive the resulting no-trade price bounds for concave utility functions and a general class of latency distributions. 

Using data from the Bitcoin market in 2018 and 2019, we show that price differences remain large during periods of high spot volatility, settlement latency, and settlement latency uncertainty. 
Cross-exchange net flows chase price differences, which indicates that market participants perceive these restrictions. 

These results shed new light on the inherent trade-off between the costs and benefits of central clearing versus blockchain-based settlement. While central clearing counterparties take on counterparty risk to guarantee instantaneous trading on non-settled positions, blockchains render trusted intermediation obsolete. However, the degree of trustworthiness for blockchain-based assets depends on the validation process's complexity, which ultimately causes settlement latency. We document the substantial economic costs of latency-related trading frictions for blockchain-based assets.

To put arbitrage capital under an exchange's custody and rely on an exchange's protection against counterparty risk by providing collateral requires the trustworthiness of the exchange. 
Though we observe an increase in trust in exchanges, measured by the increase in funds under the custody of exchanges, our results indicate that intermediation services are still insufficiently utilized to exploit cross-exchange price differences. Frequent occurrences of centralized exchange defaults are a vivid reminder of these default risks.

In fact, we demonstrate that settlement latency remains a statistically and economically significant driving force of time-varying price differences, even when we control for exchange-specific inventory holdings and margin trading possibilities. These results indicate that circumventing settlement latency via alternative strategies is not sufficiently pervasive to offset the impact of arbitrage bounds completely. A possible reason is a lack of trust in the capabilities of CEXes to serve as central counterparties.

This paper thus contributes to an ongoing debate on the organization of clearing on financial markets and the role of third-party intermediation for reliable settlement systems. Our analysis demonstrates that a decentralized system cannot easily replace central clearing. Removing the frictions (and costs) induced by third-party intermediation cause novel trading frictions with non-trivial implications for pricing. 
First, limits to arbitrage implied by settlement latency may harm price efficiency, as the lower activity of arbitrageurs reduces the information flow across exchanges. 
Second, deviations from the law of one price affect the pricing of assets, as risk-neutral probabilities are not uniquely defined. 
Third, the implied costs of settlement latency depend on the design of the blockchain and should influence the decision of whether to migrate to a decentralized settlement system. 

\section{Data availability}\label{sec:data_availability}

All code is available on \url{https://github.com/voigtstefan/building_trust_takes_time}. All data can be retrieved from \url{https://doi.org/10.5061/dryad.q2bvq83rn}.

\FloatBarrier
\newpage
\bibliography{references}

\appendix
\renewcommand{\thesection}{\hspace{-2.5ex}}
\renewcommand{\thesubsection}{\Alph{subsection}}
\setcounter{table}{0}
\renewcommand{\thetable}{\Alph{subsection}\Roman{table}}
\setcounter{figure}{0}
\renewcommand{\thefigure}{\Alph{subsection}\arabic{figure}}
\setcounter{equation}{0}
\renewcommand{\theequation}{\Alph{subsection}\arabic{equation}}
\counterwithin*{equation}{subsection}

\section{Appendix}
\subsection{Latency Distribution Under Stochastic Volatility}\label{sec:appendix_clock_change}

We can relax the assumption that $\sigma_t^s$ is constant over the interval $[t, t+\tau]$ by allowing $\sigma_t^s$ to vary over time. More specifically, let $\sigma_t^s: \mathbb{R}_+ \rightarrow \mathbb{R}_+$ with $\theta(\tau):=\int\limits_t^{t+\tau}\left(\sigma_k^s\right)^2dk <~\infty \quad\forall \tau$, i.e., the volatility of the high-price exchange follows a (deterministic) path with bounded integrated variance. Assuming  $\mu_t^s=0$, we can then rewrite the log-returns of the arbitrageur for given latency $\tau$ as
\begin{equation}\label{eq:returns_dynamic_vola}
	r^{b, s}_{(t:t+\tau)} = \delta^{b, s}_{t} + \int\limits_t ^{t+\tau} \sigma^s_{k} dW^s_k.
\end{equation}
The integral above corresponds to a Gaussian process with independent increments. More specifically, we get
\begin{equation}
	\mathbb{E}_t\left(\left(r^{b, s}_{(t:t+\tau)} - \delta^{b, s}_{t}\right)^2\right) =  \theta(\tau) -\theta(0) = \mathbb{E}_t\left(W^s_{\theta(\tau)} - W^s_{\theta(0)}\right).
\end{equation}
In other words, the time-changed Brownian motion $W^s_{\theta(t)}$ has the same distribution as the log returns given in Equation \eqref{eq:returns_dynamic_vola} \cite[e.g.,][]{Durrett.1984, Barndorff-Nielsen.2002}. We can thus rewrite the return process as
\begin{equation}
	r^{b, s}_{(t:t+\tau)} = \delta^{b, s}_{t} + \int\limits_t^{t+\theta(\tau)} dW^s_k,
\end{equation}
The implications of Lemma \ref{lemma:characteristic_function} still hold. Still, we need to compute the moment-generating function of the transformed latency $m_{\theta(\tau)}(u)$, which depends on the latency distribution and the dynamics of the volatility process.
First, note that as $\theta(\tau)$ is strictly increasing, the probability integral transformation yields the distribution of $\tau(\theta)$, 
\begin{equation}
	\mathbb{P}_t\left(\theta(\tau) = y\right) = \mathbb{P}_t\left(\tau = \theta^{-1}\left(y\right)\right) \quad\forall y >0.
\end{equation}
Finally, the distribution of $\theta(\tau)$ is fully described via its characteristic function, which is of the form
\begin{equation}\label{eq:charact_func_general_vola}
	\varphi_{\theta(\tau)}\left(u\right)=\mathbb{E}_t\left(e^{i \theta(\tau)u}\right)
	=\frac{1}{2 \pi}  \int_{0}^{\infty} \int_{-\infty}^{\infty} \varphi_\tau\left(s\right) e^{-i s \tau } ds e^{i \theta(\tau) u}  d\tau. 
\end{equation} 
L\'{e}vy's characterization allows extending these ideas to more general non-deterministic integrands and stochastic time changes. 
Although Equation \eqref{eq:charact_func_general_vola} allows deriving theoretical arbitrage bounds based on Theorem~\ref{theorem:general} for every continuous local martingale, we restrict our analysis to analytically more tractable and intuitive dynamics of the price process and the associated settlement latency.  

\subsection{Proofs}\label{sec:appendix_proofs}

\begin{proof}[Proof of Lemma~\ref{lemma:characteristic_function}.]
	The proof of the lemma is an application of Equation (2.2) in \cite{Barndorff-Nielsen.1982}. 
\end{proof}

\begin{proof}[Proof of Theorem \ref{theorem:general}.]
	First, note that the characteristic function in Lemma \ref{lemma:characteristic_function} yields the first moment $\mu_r$ of the returns as
	\begin{align}\label{eq:expectation}
	\mathbb{E}_t\left(r_{(t:t+\tau)}^{b, s}\right) &= (-i)\frac{\partial}{\partial u}\varphi_{r_{(t:t+\tau)}^{b, s}}\left(u\right)\bigg\rvert_{u=0} \nonumber\\
	& = \delta_t^{b,s}e^{iu\delta_t^{b,s}}m_{\tau}\left(iu\mu^s_t-\frac{1}{2}u^2(\sigma^s_{t})^2\right)  \nonumber\\
	& \qquad + e^{iu\delta_t^{b,s}}m'_{\tau}\left(iu\mu^s_t-\frac{1}{2}u^2(\sigma^s_{t})^2\right)\left(\mu_t^s+iu(\sigma_t^s)^2\right)\bigg\rvert_{u=0} \nonumber\\
	& = \delta_t^{b,s}+ \mathbb{E}_t(\tau)\mu^s_t,
	\end{align}
	since $m_{\tau}(0)=1$ and $m'_{\tau}(0)=\mathbb{E}_t(\tau)$ by definition of the moment generating function.	
	
	In the spirit of \cite{Arditti.1967} and \cite{Scott.1980}, we express the expected utility of the arbitrageur by a Taylor expansion which results in a function of the higher-order moments of the return distribution. A Taylor expansion of a general utility function $U_\gamma(r)$ around the mean $\mu_r$ yields 
	\begin{align}
	U_\gamma\left(r_{(t:t+\tau)}^{b, s}\right) &= \sum\limits_{k=0}^\infty\frac{U_\gamma^{(k)}\left(\mu_r\right)}{k!}\left(r_{(t:t+\tau)}^{b, s} - \mu_r\right)^k,
	\end{align}
	where $U_\gamma^{(k)}\left(\mu_r\right) := \frac{\partial^k}{\partial\mu_r^k}U_\gamma\left(\mu_r\right)$.	Then, taking expectations yields
	\begin{align}\label{equation:CE_taylor_full}
	\mathbb{E}_t\left(U_\gamma\left(r_{(t:t+\tau)}^{b, s}\right)\right) = U_\gamma\left(\mu_r\right) +  \sum\limits_{k=2}^\infty\frac{U_\gamma^{(k)}\left(\mu_r\right)}{k!}\mathbb{E}_t\left(\left(r_{(t:t+\tau)}^{b, s} - \mu_r\right)^k\right).
	\end{align}
	Following \cite{Markowitz.1952}, we next consider a first-order Taylor expansion for the CE. We thus implicitly assume that the risk premium, $\mu_r-CE$, is small and that higher-order moments vanish:
	\begin{align}\label{equation:CE_taylor1} 
	\mathbb{E}_t\left(U_\gamma\left(r_{(t:t+\tau)}^{b, s}\right)\right) = U_\gamma\left(CE\right) = U_\gamma\left(\mu_r\right) + U_\gamma'\left(\mu_r\right)\left(CE - \mu_r\right).
	\end{align} 
	Moreover, the first-order Taylor expansion provides a convenient closed-form approximation of the certainty equivalent, which is linear in the moments of the return distribution. We obtain the equation in the theorem by equating \eqref{equation:CE_taylor_full} and \eqref{equation:CE_taylor1}, plugging in \eqref{eq:expectation} and solving for $CE$.
\end{proof}

\begin{proof}[Proof of Lemma \ref{lemma:optimal_strategy_iso}.]
	The proof follows directly from applying Theorem \ref{theorem:general} together with the derivatives of the isoelastic utility function, which yields
	\begin{equation}
	d_t^s -  \frac{1}{2}\frac{\gamma}{d_t^s}(\sigma_t^s)^2\mathbb{E}_t\left(\tau\right) - \frac{1}{8} \frac{\gamma(\gamma+1)(\gamma+2)}{(d_t^s)^3} (\sigma_t^s)^4\mathbb{E}_t\left(\tau^2\right) = 0.
	\end{equation}	
	Then, by Descartes' rule of signs, there is exactly one positive real root to the polynomial
	\begin{equation}
	(d_t^s)^4 - \frac{1}{2}\gamma(\sigma_t^s)^2\mathbb{E}_t\left(\tau\right)(d_t^s)^2 - \frac{1}{8} \gamma(\gamma+1)(\gamma+2) (\sigma_t^s)^4\mathbb{E}_t\left(\tau^2\right) = 0.
	\end{equation}
	All four solutions of the quartic polynomial are given by 
	\begin{equation}
	d_t^s=\pm\frac{1}{\sqrt{2}}\sqrt{\frac{\gamma}{2}(\sigma^s_t)^2\mathbb{E}_t\left(\tau\right)\pm\sqrt{\frac{\gamma^2}{4}(\sigma^s_t)^4\mathbb{E}_t\left(\tau\right)^2+\frac{\gamma(\gamma+1)(\gamma+2)}{2}(\sigma^s_t)^4\mathbb{E}_t\left(\tau^2\right)}}.
	\end{equation}
	However, since
	\begin{equation}
	\frac{\gamma}{2}(\sigma^s_t)^2\mathbb{E}_t\left(\tau\right)<\sqrt{\frac{\gamma^2}{4}(\sigma^s_t)^4\mathbb{E}_t\left(\tau\right)^2+\frac{\gamma(\gamma+1)(\gamma+2)}{2}(\sigma^s_t)^4\mathbb{E}_t\left(\tau^2\right)}
	\end{equation}
	holds for all $\gamma>0$, $\sigma^s_t>0$ and $\mathbb{E}_t\left(\tau^2\right)>0$, the expression in the lemma gives the unique positive real root.
\end{proof}

\begin{proof}[Proof of Lemma \ref{lemma:trans_cost}.]

The Taylor representation of $U_\gamma(\tilde{r})$ yields for $\rho^* :=\log\left(\frac{1 + \rho^{b, A}(q)}{1 - \rho^{s, B}(q)}\right)$: 
	\begin{align}
		\mathbb{E}_t\left(U_\gamma(\tilde{r})\right) &= \delta_t^{b, s} + \mathbb{E}_t(\tau)\mu^s_t - \rho^*   \nonumber\\
		&\qquad+\sum\limits_{k=2}^\infty\frac{U_\gamma^{(k)}\left(\delta_t^{b, s} + \mathbb{E}_t(\tau)\mu^s_t - \rho^* \right) }{k!U'_\gamma\left(\delta_t^{b, s} + \mathbb{E}_t(\tau)\mu^s_t - \rho^* \right)}\mathbb{E}_t\left(\left(r_{(t:t+\tau)}^{b, s}  - \rho^* - \delta_t^{b, s} - \mathbb{E}_t(\tau)\mu^s_t\right)^k\right).
	\end{align}
	Let $d_t^s$ be the arbitrage bound (in absence of transaction costs) as defined in Equation~\eqref{definition:ce_root}. Then, $d_t^s +\ln\left(\frac{1 + \rho_t^{b, A}(q)}{1 - \rho_t^{s, B}(q)}\right)$ is a root of the function
	\begin{align}
		\tilde F(d) :=& d + \mathbb{E}_t(\tau)\mu^s_t - \rho^*  \nonumber\\
		&\qquad+\sum\limits_{k=2}^\infty\frac{U^{(k)}_\gamma\left(d+ \mathbb{E}_t(\tau)\mu^s_t- \rho^* \right) }{k!U'_\gamma\left(d+ \mathbb{E}_t(\tau)\mu^s_t -\rho^* \right)}\mathbb{E}_t\left(\left(r_{(t:t+\tau)}^{b, s} - \rho^* - d - \mathbb{E}_t(\tau)\mu^s_t\right)^k\right).
	\end{align} 
	Therefore, $\mathbb{E}_t\left(U_\gamma(\tilde{r})\right)$ is positive if and only if 
	\begin{equation}
		\delta_t^{b,s} > d_t^s +\ln\left(\frac{1 + \rho_t^{b, A}(q)}{1 - \rho_t^{s, B}(q)}\right).
	\end{equation}
\end{proof}

\begin{proof}[Proof of Lemma \ref{lemma:settlement_fees}.]
	The proof directly follows from Lemma \ref{lemma:trans_cost} and Theorem \ref{theorem:general}.
\end{proof}

\begin{proof}[Proof of Lemma \ref{lemma:optimal_q_f}.]
	We cast the arbitrageur's optimization problem in terms of the Lagrangian
	\begin{align}
		\mathcal{L}(q,f;\xi)= &B_t^s(1-\rho^{s,B}(q))q + A_t^b(1+\rho^{b,A}(q+f))(q+f) \nonumber\\
		&\qquad - \xi\left(d_t^s(f) - \delta_t^{b,s}+ \log\left(1 + \rho^{b, A}(q)\right) - \log\left(1 - \rho^{s, B}(q)\right)\right)
	\end{align}
	and observe that the corresponding Karush-Kuhn-Tucker (KKT) conditions imply
	\begin{align}
		q = 0 \quad\vee & \quad B_t^s\left((1-\rho^{s,B}(q))-\rho^{s,B'}(q)q\right)  \nonumber\\
		& \qquad - A_t^b\left((1+\rho^{b,A}(q+f))+\rho^{b,A'}(q+f)(q+f)\right) \nonumber\\
		& \qquad -\xi\left(\frac{\rho^{b,A'}(q+f)}{1 + \rho^{b,A}(q+f)}-\frac{\rho^{s,B'}(q)}{1 + \rho^{s,B}(q)}\right) = 0 \label{eq:kkt1}\\
		f = 0  \quad\vee & \quad - A_t^b\left((1+\rho^{b,A}(q+f))+\rho^{b,A'}(q+f)(q+f)\right) \nonumber \\
		& \qquad - \xi\left(\frac{d}{df}d_t^s(f) + \frac{\rho^{b,A'}(q+f)}{1 + \rho^{b,A}(q+f)}\right) = 0 \label{eq:kkt2}\\
		\xi = 0  \quad\vee & \quad d_t^s(f) - \delta_t^{b,s} \nonumber\\
		&\qquad + \log\left(1 + \rho^{b, A}(q+f)\right) - \log\left(1 - \rho^{s, B}(q)\right) = 0, \label{eq:kkt3}
	\end{align}
	We first consider the case of $\xi = 0$. Conditions \eqref{eq:kkt1} and \eqref{eq:kkt2} now become
	\begin{align}
		q = 0 \quad\vee & \quad B_t^s\left((1-\rho^{s,B}(q))-\rho^{s,B'}(q)q\right) \nonumber \\
		&\qquad - A_t^b\left((1+\rho^{b,A}(q+f))+\rho^{b,A'}(q+f)(q+f)\right) = 0 \\
		f = 0  \quad\vee & \quad - A_t^b\left((1+\rho^{b,A}(q+f))+\rho^{b,A'}(q+f)(q+f)\right) = 0
	\end{align}
	which only holds if 
	\begin{align}
		1+\rho^{b,A}(q+f)=-\rho^{b,A'}(q+f)(q+f).
	\end{align}
	Since $\rho^{b,A'}(q+f)> 0$ by Assumption~\ref{assumption:transaction_costs}, this cannot be the case for any $q>0$ or $f>0$. Also, note that $\xi=q=f=0$ implies a contradiction. Therefore, the constraint \eqref{eq:constraint} cannot be slack at the optimum, and there does not exist a candidate solution for $\xi=0$. 
	
	Next, we turn to the analysis of $\xi > 0$. The simple case of $q=0$ does not deliver any positive returns, and it does not make sense for the arbitrageur to pay any fee $f>0$. If anything, the arbitrageur would prefer not to trade at all, i.e., $q=f=0$. We are left with the two interesting cases of $q>0$. 
	
	For $f=0$, the KKT conditions give the candidate solution $\{q_1, f_1, \xi_1\}$ as solutions to the system of equations
	\begin{align}
		B_t^s\left((1-\rho^{s,B}(q_1))-\rho^{s,B'}(q_1)q_1\right) - A_t^b\left((1+\rho^{b,A}(q_1))+\rho^{b,A'}(q_1)(q_1)\right) \qquad& \nonumber\\
		-\xi_1\left(\frac{\rho^{b,A'}(q_1)}{1 + \rho^{b,A}(q_1)}-\frac{\rho^{s,B'}(q_1)}{1 + \rho^{s,B}(q_1)}\right) &= 0 \\
		d_t^s(f_1) - \delta_t^{b,s}+ \log\left(1 + \rho^{b, A}(q_1)\right) - \log\left(1 - \rho^{s, B}(q_1)\right) &= 0 \\
		f_1 &= 0.
	\end{align}
	For $f>0$, we can get the candidate solution $\{q_2, f_2, \xi_2\}$ as solutions to 
	\begin{align}
		B_t^s\left((1-\rho^{s,B}(q_2))-\rho^{s,B'}(q_2)q_2\right)  \qquad& \nonumber\\
		- A_t^b\left((1+\rho^{b,A}(q_2+f_2))+\rho^{b,A'}(q_2+f)(q_2+f_2)\right)\qquad& \nonumber\\
		-\xi\left(\frac{\rho^{b,A'}(q_2+f_2)}{1 + \rho^{b,A}(q_2+f_2)}-\frac{\rho^{s,B'}(q_2)}{1 + \rho^{s,B}(q_2)}\right) &= 0\label{eq:candidate1}\\
		- A_t^b\left((1+\rho^{b,A}(q_2+f_2))+\rho^{b,A'}(q_2+f_2)(q_2+f_2)\right) \qquad \nonumber \\
		- \xi\left(\frac{d}{df}d_t^s(f_2) + \frac{\rho^{b,A'}(q_2+f_2)}{1 + \rho^{b,A}(q_2+f_2)}\right) &=0 \label{eq:candidate2}\\
		d_t^s(f_2) - \delta_t^{b,s}+ \log\left(1 + \rho^{b, A}(q_2+f_2)\right) - \log\left(1 - \rho^{s, B}(q_2)\right) &= 0.
	\end{align}
	However, combining \eqref{eq:candidate1} and \eqref{eq:candidate1} shows that the solutions are only admissible if
	\begin{align}\label{eq:fee_conditions}
		\xi = \frac{B_t^s\left((1-\rho^{s,B}(q_2))-\rho^{s,B'}(q_2)q_2\right)}{\frac{d}{df}d_t^s(f_2) - \frac{\rho^{s,B'}(q_2)}{1 + \rho^{s,B}(q_2)}}>0.
	\end{align}
	Equation~\eqref{eq:fee_conditions} now provides us with the necessary conditions for a solution to the problem that entails a strictly positive settlement fee. Namely, $q_2>0$, $f_2>0$ $\xi_2>0$ can only be a solution if one of the following two conditions holds
	\begin{itemize}
		\item[(i)] $-\frac{d}{df}d_t^s(f_2) > \frac{\rho^{s,B'}(q_2)}{1 - \rho^{s,B}(q_2)}$ and $1-\rho^{s,B}(q_2)>\rho^{s,B'}(q_2)q_2$
		\item[(ii)] $-\frac{d}{df}d_t^s(f_2) < \frac{\rho^{s,B'}(q_2)}{1 - \rho^{s,B}(q_2)}$ and $1-\rho^{s,B}(q_2)<\rho^{s,B'}(q_2)q_2$.
	\end{itemize}
	However, condition (ii) cannot hold at the maximum since $1-\rho^{s, B}(q_2)<\rho^{s, B'}(q_2)q_2$ means that the trading quantity is such that the marginal price impact exceeds the average price impact. In this case, the arbitrageur would reduce the trading quantity to raise her total return. Consequently, (i) remains the necessary condition for a candidate solution with a positive settlement fee, which completes the proof.
\end{proof}

\subsection{Robustness Tests}\label{sec:appendix_robustness}

\begin{table}[ht]
	\centering
	\caption{Price differences, ratings, currency controls, and fiat stability}\label{tab:robust_regions}
 \resizebox{\columnwidth}{!}{%

\centering
\begin{tabular}{lccccccc}
      \toprule
   \emph{Dependent Variable} & \multicolumn{7}{c}{Price Differences (in \%)} \\ 
   \cmidrule(lr){2-2}\cmidrule(lr){3-5}\cmidrule(lr){6-8}
   Subset                                                    &                & USA           & Europe        & All           & USDT          & USD           & All \\   
                                                       & (1)            & (2)           & (3)           & (4)           & (5)           & (6)           & (7)\\  
   \midrule
   Arbitrage Bound (in \%)                                   & 4.768$^{***}$  & 1.805$^{***}$ & 5.269$^{***}$ & 2.738$^{***}$ & 2.711$^{***}$ & 2.946$^{***}$ & 2.922$^{***}$\\   
                                                             & (16.81)        & (4.326)       & (4.486)       & (5.795)       & (7.278)       & (8.313)       & (9.922)\\   
   Arbitrage Bound (in \%) $\times$ AA Rating     & -3.199$^{***}$ &               &               &               &               &               &   \\   
                                                             & (-7.690)       &               &               &               &               &               &   \\   
      \midrule
Controls: Inventory                           & Yes            & Yes           & Yes           & Yes           & Yes           & Yes           & Yes\\  
   Exchange Pair FE                               & Yes            & Yes           & Yes           & Yes           & Yes           & Yes           & Yes\\  
   Observations                                              & 1,089,949      & 139,319       & 78,416        & 217,735       & 175,072       & 512,818       & 687,890\\  
   \bottomrule
   \end{tabular}
}
\subcaption*{\emph{Notes:} This table provides OLS estimates based on a regression of hourly average exchange pair price differences and the main components of price risk due to settlement latency. \emph{Price Difference} is the average hourly exchange pair price difference (in percent).
\emph{AA rating} is a dummy variable that is one if the exchange received a top rating from \emph{cryptocompare} on November 1st, 2019, and zero otherwise.
The regional subsets (Columns (2) - (4)) correspond to all exchange pairs where both exchanges are active within the same region (in percent, only \emph{USA}, only \emph{Europe}, \emph{all} contains all exchange pairs where both exchanges are active within the same region).
\emph{USDT} corresponds to the subset of all exchange pairs where both exchanges denote BTC versus Tether, \emph{USD} corresponds to the subset of exchange pairs where both exchanges denote BTC versus USD. \emph{All} corresponds to the subset of exchange pairs where both exchanges denote BTC against the same fiat currency (USD or USDT). 
\emph{Arbitrage bound} corresponds to the average hourly high-price exchange estimated arbitrage bound. 
Within the regressions, we control for \emph{inventory} (the average number of Bitcoins controlled by all wallets associated with the high-price exchange at hour $t-1$). 
We report $t$-statistics based on standard errors following \cite{Newey.1987} using 24 lags in parentheses. 
$^{***}, ^{**}$, and $^{*}$ indicate statistical significance on the $1\%, 5\%$ and $10\%$ levels (two-tailed), respectively.}
\end{table}
\end{document}